\def \VersionLong {}
\def \VersionAuthor {}
\ifdefined\VersionAuthor
	\newcommand{\AuthorVersion}[1]{#1}
	\newcommand{\FinalVersion}[1]{}
\else
	\newcommand{\AuthorVersion}[1]{}
	\newcommand{\FinalVersion}[1]{#1}
\fi

\documentclass[runningheads,a4paper,10pt,envcountsame]{llncs}

\usepackage{comment}
\ifdefined\VersionAuthor
\includecomment{AuthorVersionBlock}
\else
\excludecomment{AuthorVersionBlock}
\fi
\includecomment{lncs}
\begin{lncs}
 \makeatletter
 \AtBeginDocument{%
   \@ifpackageloaded{hyperref}
   {\def\@doi#1{\href{https://doi.org/#1}
       {\ttfamily https://doi.org/#1}\egroup}}
   {\def\@doi#1{\ttfamily https://doi.org/#1\egroup}}
   \def\doi{\bgroup\catcode`\_=12\relax\@doi}}
\makeatother
\end{lncs}
\begin{lncs}
\makeatletter
\def\@biblabel#1{[#1]}

\makeatother
\end{lncs}

\usepackage[utf8]{inputenc}
\usepackage{csquotes}
\usepackage{lmodern} %
\usepackage[T1]{fontenc} %

\usepackage[svgnames,table]{xcolor}
\definecolor{USPNcobalt}{HTML}{293358}
\definecolor{USPNocre}{HTML}{8b7d6d}
\definecolor{USPNblanc}{HTML}{ffffff}
\definecolor{USPNceruleen}{HTML}{354878}
\definecolor{USPNsable}{HTML}{ad947e}

\begin{lncs}
 \usepackage[
		pdfauthor={Masaki Waga, Étienne André},%
		pdftitle={Hyper pattern matching},
		breaklinks  = true,
		colorlinks  = true,
	\ifdefined \VersionWithComments
	\fi
		citecolor   = USPNsable,
		linkcolor   = USPNocre,
		urlcolor    = USPNceruleen,
	]{hyperref}
\end{lncs}

\usepackage[ruled,lined,linesnumbered,noend]{algorithm2e}

\DontPrintSemicolon{}
\SetKwInOut{Input}{Input}
\SetKwInOut{Output}{Output}
\SetKw{KwRemove}{remove}
\SetKw{KwPush}{add}
\SetKw{KwPop}{pop}
\SetKw{KwFrom}{from}
\SetKw{KwTo}{to}
\SetKw{KwCompute}{compute}
\SetKw{KwReturn}{return}
\SetKw{KwBreak}{break}
\SetKw{KwContinue}{continue}
\SetKw{KwPick}{pick}
\SetKw{KwLet}{let}
\SetKw{KwBe}{be}
\SetKw{KwSplit}{split}
\SetKw{KwInto}{into}
\SetKw{KwFind}{find}
\SetKw{KwFeed}{feed}
\SetKwProg{Fn}{Function}{:}{}
\SetAlgorithmName{Algorithm}{algorithm}{List of Algorithms}

\SetCommentSty{mycommfont}
\SetNlSty{myLNfont}{}{} %
\SetKwSty{myKWfont}

\newcommand{\inlineAlgoStyleKW}[1]{\textbf{#1}}

\usepackage{subcaption}

\usepackage{amsmath} %
\usepackage{amssymb} %
\usepackage{mathtools}
\usepackage{booktabs}
\usepackage[misc,geometry]{ifsym} %
\usepackage[range-phrase={--},list-final-separator={, and }]{siunitx}
\usepackage{pifont}%
\newcommand{\xmark}{\ding{55}}%

\usepackage{thm-restate}

\usepackage{enumitem} %

\newlist{ienumerate}{enumerate*}{1}
\setlist[ienumerate]{label=\textit{\roman*})}

\newlist{oneenumerate}{enumerate*}{1}
\setlist[oneenumerate]{label=\arabic*)}

\ifdefined \VersionLong
	\newcommand{\LongVersion}[1]{#1}
	\newcommand{\ShortVersion}[1]{}
\else
	\newcommand{\LongVersion}[1]{}
	\newcommand{\ShortVersion}[1]{#1}
\fi

\ifdefined\VersionAuthor
	\usepackage[backend=biber,backref=true,style=alphabetic,url=false,doi=true,defernumbers=true,sorting=anyt,maxnames=99]{biblatex} %
	\DeclareSourcemap{
		\maps[datatype=bibtex]{
			\map[overwrite]{
				\step[fieldsource=toappear, match={true},
					fieldset=note, fieldvalue={To appear.~{}}, append]
			}
		}
	}

	\renewbibmacro*{doi+eprint+url}{%
		\iftoggle{bbx:doi}
			{\color{black!40}\footnotesize\printfield{doi}}
			{}%
		\newunit\newblock
		\iftoggle{bbx:eprint}
			{\usebibmacro{eprint}}
			{}%
		\newunit\newblock
		\iftoggle{bbx:url}
			{\usebibmacro{url+urldate}}
			{}%
	}

\fi
\ifdefined\VersionWithComments

\usepackage{everypage}

\AddEverypageHook{%
	\begin{tikzpicture}[remember picture, overlay]
		\node[
		align=center,
		anchor=north east,
		xshift=-1cm,
		yshift=0cm,
		text=red
		] at (current page.north east) {\scalebox{2}{\bfseries\LARGE Draft!}\\{\tiny\today{}}};
	\end{tikzpicture}
}

\fi
\usepackage[capitalise,english,nameinlink]{cleveref} %
\crefname{line}{\text{line}}{\text{lines}} %

\newcommand{\fakeParagraph}[1]{\paragraph{#1}}

\usepackage{tikz}
\usetikzlibrary{arrows,automata,calc,positioning}
\tikzstyle{NFA}=[auto, ->, >=stealth']
\tikzstyle{every node}=[initial text=]
\tikzstyle{state}=[circle, minimum size=12pt, draw=black, fill=blue!10, inner sep=1pt]
\tikzstyle{NFApath}=[-stealth, thick]
\tikzstyle{final}=[accepting, fill=blue!50]

\definecolor{coloract}{rgb}{0, 0.3, 0}
\definecolor{colorloc}{rgb}{0.4, 0.4, 0.65}

\newcommand{\styleact}[1]{\ensuremath{\textcolor{coloract}{{\mathit{#1}}}}}
\newcommand{\styleskippedact}[1]{\ensuremath{\textcolor{gray}{{\mathit{#1}}}}}
\newcommand{\styleactP}[1]{\ensuremath{\textcolor{red!50!black}{{\mathit{#1}}}}}
\newcommand{\styleactQ}[1]{\ensuremath{\textcolor{green!50!black}{{\mathit{#1}}}}}

\newcommand{\defProblem}[3]
{%
\noindent\fcolorbox{black}{USPNsable!15}{
	\begin{minipage}{.95\columnwidth}
		\textbf{#1 problem:}\\
		\textsc{Input}: #2\\
		\textsc{Output}: #3
	\end{minipage}
}

	\smallskip

}

\tikzstyle{rqanswer} = [
 draw=black,
 fill=gray!30,
 text=black,
 line width=0.5pt,
  text width = \linewidth - 1.6 ex - 1pt,
  inner sep = 0.8 ex,
  rounded corners=4pt]

\newcommand{\cellTimeout}{{\color{red}\textbf{T/O}}}

\usepackage{verbatim} %

\newcommand{\init}{_0}

\newcommand{\setN}{{\mathbb N}}
\newcommand{\setNplus}{\setN_{+}}

\newcommand{\powerset}[1]{\mathcal{P} (#1)}
\newcommand{\partfun}{\nrightarrow}
\newcommand{\domain}{\mathrm{dom}}
\newcommand{\projectsymbol}{\pi}
\newcommand{\project}[2]{\projectsymbol({#1}, {#2})}
\newcommand{\cardinality}[1]{\ensuremath{\lvert #1 \rvert}}
\newcommand{\tuple}[1]{\ensuremath{\langle #1 \rangle}}
\newcommand{\complexityClass}{\ensuremath{\mathcal{O}}}
\newcommand{\AP}{\mathit{AP}}
\newcommand{\prop}[1][]{p_{#1}}
\newcommand{\fml}[1][]{\varphi_{#1}}
\newcommand{\extendedword}[1][]{\overline{w}^{#1}}
\newcommand{\word}[1][]{w^{#1}}
\newcommand{\wordi}[2][]{w^{#1}_{#2}}
\newcommand{\worda}[1][]{v^{#1}} %
\newcommand{\wordb}[1][]{u^{#1}} %
\newcommand{\wordSet}{\mathbf{w}}
\newcommand{\startend}{\ensuremath{\$}} %
\newcommand{\slice}[3]{{#1}|_{[{#2}, {#3}]}}
\newcommand{\Alphabet}[1][]{\Sigma}
\newcommand{\emptyword}{\varepsilon}

\newcommand{\NFA}{\mathcal{A}}
\newcommand{\loc}{s}
\newcommand{\Loc}{S}
\newcommand{\InitLoc}{\Loc_{0}}
\newcommand{\initLoc}{\loc_{\init}}
\newcommand{\Final}{F}
\newcommand{\FinalLoc}{\Loc_{\Final}}
\newcommand{\Edges}{\Delta}

\newcommand{\Lg}{\mathcal{L}}
\newcommand{\silentaction}{\ensuremath{\varepsilon}}
\newcommand{\action}{\ensuremath{\sigma}}

\newcommand{\loci}[1]{\ensuremath{\loc_{#1}}}
\newcommand{\locfinal}{\ensuremath{\loci{f}}}
\newcommand{\asNFA}[1]{\overline{#1}}
\newcommand{\NAA}{\mathsf{A}}
\newcommand{\hyperLg}{\mathfrak{L}}
\newcommand{\Vars}{K}
\newcommand{\var}[1][l]{{#1}}

\newcommand{\assign}[3]{#1[#2 \leftarrow #3]} %
\newcommand{\Match}{\mathcal{M}}

\newcommand{\waitingQueue}{\mathcal{Q}}
\newcommand{\currentConfigurations}{\mathcal{C}}
\newcommand{\reachedLoc}{\mathcal{R}}
\newcommand{\SkipQS}{\Delta_{\mathrm{QS}}}
\newcommand{\LastQS}[1]{\Lambda_{\mathrm{QS}}^{#1}}

\newcommand{\ShortestMatching}[1]{\mathcal{SM}^{#1}}
\newcommand{\SkipKMP}{\Delta_{\mathrm{KMP}}}
\newcommand{\configuration}{c}
\newcommand{\stylealgo}[1]{\ensuremath{\mathcal{#1}}}
\newcommand{\algoPM}{\ensuremath{\stylealgo{PM}}}
\newcommand{\algoHPM}{\ensuremath{\stylealgo{HPM}}}
\newcommand{\algoHPMFJS}{\ensuremath{\algoHPM^{\mathit{FJS}}}}
\newcommand{\algoHPMP}{\ensuremath{\stylealgo{HPM}_P}}
\newcommand{\algoHPMFJSP}{\ensuremath{\stylealgo{HPM}^{\mathit{FJS}}_P}}
\newcommand{\pattern}{\word_{p}}
\newcommand{\target}{\word}
\newcommand{\stylecomplexity}[1]{\textsf{#1}}
\newcommand{\NP}[1]{\stylecomplexity{NP}}
\newcommand{\NPSPACE}[1]{\stylecomplexity{NPSPACE}}
\newcommand{\PSPACE}[1]{\stylecomplexity{PSPACE}}
\newcommand{\ourTool}{\textsc{HypPAu}}
\newcommand{\interference}{\textsc{Interference}}
\newcommand{\robustness}{\textsc{Robustness}}
\newcommand{\packetPair}{\textsc{PacketPairs}}
\newcommand{\manydimensions}{\textsc{ManyDirs}}

\newcommand{\DFA}{\mathcal{A}}
\newcommand{\eg}{e.g.,\xspace}

\newcommand{\ie}{i.e.,\xspace}

\newcommand{\st}{s.t.\xspace}
\newcommand{\wrtwithspace}{w.r.t.~} %
\newcommand{\resp}{resp.\xspace}

\begin{lncs}
\usepackage{orcidlink}
\renewcommand{\orcidID}[1]{\orcidlink{#1}} 
\end{lncs}

\title{Hyper pattern matching} %

\begin{lncs}
 \author{
 \ifdefined\VersionAnonymous%
 \else
 Masaki Waga\orcidID{0000-0001-9360-7490}\inst{1,2}\and 
 Étienne André\orcidID{0000-0001-8473-9555}\inst{3,4}
 \LongVersion{\thanks{%
    This is the author (and extended) version of the manuscript of the same name published in the proceedings of the 25th International Conference on Runtime Verification (RV 2025).
    The final version is available at \url{www.springer.com}.
    }%
}
\fi
}
\institute{%
\ifdefined\VersionAnonymous%
\else
 Graduate School of Informatics, Kyoto University, Kyoto, Japan
\and
 National Institute of Informatics, Tokyo, Japan
\and
 Université Sorbonne Paris Nord, CNRS, Laboratoire d'Informatique de Paris Nord, LIPN, F-93430 Villetaneuse, France %
\and
 Institut universitaire de France (IUF)
\fi
}
\end{lncs}

\begin{document}
\sloppy

\maketitle{}
\setcounter{footnote}{0}

\begin{abstract}
    In runtime verification, pattern matching, which searches for occurrences of a specific pattern within a word, provides more information than a simple violation detection of the monitored property, by locating concrete evidence of the violation.
    However, witnessing violations of some properties, particularly \emph{hyperproperties}, requires evidence across multiple input words or different parts of the same word, which goes beyond the scope of conventional pattern matching.
    We propose here \emph{hyper pattern matching}, a generalization of pattern matching over a set of words.
    Properties of interest include robustness and (non-)interference.
    As a formalism for patterns,
        we use nondeterministic asynchronous finite automata (NAAs).
    We first provide a naive algorithm for hyper pattern matching and then devise several heuristics for better efficiency.
    Although we prove the \NP{}-completeness of the problem, our implementation \ourTool{} is able to address several case studies scalable in the length, number of words (or logs) and number of dimensions, suggesting the practical relevance of our approach.
	\keywords{runtime verification \and hyperproperties \and pattern matching.}
\end{abstract}
\section{Introduction}\label{section:introduction}
Runtime verification is a lightweight formal method that focuses on monitoring and analyzing system executions (or logs) to ensure they comply with desired specifications.
Pattern matching consists of searching for occurrences of a specific \emph{pattern} (such as a sequence of symbols or a regular expression) within a word or log.
Many important system requirements, such as noninterference, symmetry, and information flow control, cannot be expressed as trace properties alone:
witnessing violations of such requirements requires evidence across \emph{multiple} input words or \emph{different parts} of the same word---which goes beyond the scope of conventional pattern matching.
For example, one may want to detect occurrences of $n$ ``$a$''s in one word \emph{and} the same number~$n$ of ``$b$''s in another word.
Or, when monitoring certain activities of a network bus, one may detect sequences of packets that are of the same size but serve two different purposes (\eg{} requests and responses) and may potentially be interwoven.

To this end, we introduce \emph{hyper pattern matching}, the process of searching for occurrences of patterns involving multiple words (or multiple portions of the same words) within a set of words (or logs).
Hyper pattern matching can extract\LongVersion{ concrete} evidence of violation of \emph{hyperproperties}~\cite{FHST19},
a generalization of trace properties that describe \emph{sets of sets} of execution traces, rather than\LongVersion{ just} sets of traces.

To represent patterns in hyper pattern matching,
    we use nondeterministic asynchronous finite automata (NAAs)~\cite{GMO21} as an extension of finite-state automata with ``directions'', that are assigned words from a set of words.
We show that NAAs are rich enough to represent violations of interesting security properties, \eg{} noninterference and robustness.
\begin{example}[counting]\label{example:intro:counting}
	Consider the NAA in \cref{figure:example}, with two \emph{directions} $\var[1]$ and~$\var[2]$; directions are assigned a (sub)word, and can be seen informally as variables ``reading'' letters in a given subword\LongVersion{ (see \cref{section:preliminaries} for a formal definition)}.
	This NAA defines pairs of words such that both words start with a ``$\startend$'' (transition from $\loci{0}$ to~$\loci{1}$ for the first word, and from $\loci{1}$ to~$\loci{2}$ for the second word), followed by the same number of ``$a$''s in the first word as of~``$b$''s in the second word (loop over $\loci{2}$ and~$\loci{3}$).
	Finally, the first word must end with a~``$\startend$''.
	Let us consider pattern matching, with an input singleton word $\wordSet = \{ \word \}$,
	with $\word = d\startend{}aa\startend{}bbb\startend{}aaa\startend{}$.
	The match set $\Match(\NAA, \wordSet)$ is
	\(\big\{ \tuple{ (\word, 2, 5 ) , (\word, 5, 7 )} , \tuple{ (\word, 9, 13 ) , (\word, 5, 8 )} \big\} \), \ie{} two pairs of two subwords, where $(\word, 2, 5 )$ denotes the subword made of the 2nd to the 5th letter of~$\word$.
\end{example}

\newcommand{\exNIone}{\ensuremath{\top}}
\newcommand{\exNIzero}{\ensuremath{\bot}}

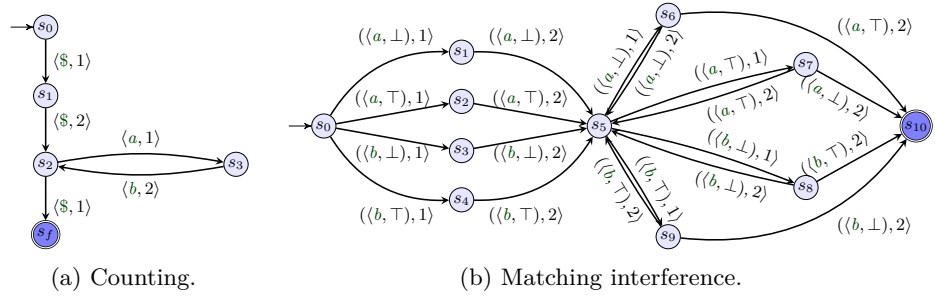
\begin{figure}[tb]
	\centering
	\begin{subfigure}[b]{0.27\textwidth}
        \scalebox{0.725}{\begin{tikzpicture}[NFA, scale=1, yscale=1, node distance=0.8cm]
        \node[state, initial] (init) {$\loci{0}$};
        \node[state, below=of init] (p1) {$\loci{1}$};
        \node[state, below=of p1] (p2) {$\loci{2}$};
        \node[state, node distance=3.0cm, right=of p2] (p3) {$\loci{3}$};
        \node[state, below=of p2, final] (final) {$\locfinal$};

        \path [NFApath]
            (init) edge node {$\tuple{\styleact{\startend}, \var[1]}$} (p1)
            (p1) edge node[pos=0.3] {$\tuple{\styleact{\startend}, \var[2]}$} (p2)
            (p2) edge[bend left=10] node[] {$\tuple{\styleact{a}, \var[1]}$} (p3)
            (p3) edge[bend left=10] node[] {$\tuple{\styleact{b}, \var[2]}$} (p2)
            (p2) edge node[pos=0.7] {$\tuple{\styleact{\startend}, \var[1]}$} (final)
        ;
        \end{tikzpicture}}
        \caption{Counting.}%
        \label{figure:example}
    \end{subfigure}
	\hfill{}
	\begin{subfigure}[b]{0.7\textwidth}
    	\scalebox{0.725}{%
        \begin{tikzpicture}[NFA]
        \node[state, initial] (init) at (0, 0) {$\loc_0$};
        \node[state] (init_a_0) at (1*2.5, 1.5*0.9) {$\loc_1$};
        \node[state] (init_a_1) at (1*2.5, 0.5*0.9) {$\loc_2$};
        \node[state] (init_b_0) at (1*2.5, -0.5*0.9) {$\loc_3$};
        \node[state] (init_b_1) at (1*2.5, -1.5*0.9) {$\loc_4$};
        \node[state] (consistent) at (2*2.5, 0*0.9) {$\loc_5$};
        \node[state] (consistent_a_0) at (2.5*2.5, 2.25*0.9) {$\loc_6$};
        \node[state] (consistent_a_1) at (3.5*2.5, 1.25*0.9) {$\loc_7$};
        \node[state] (consistent_b_0) at (3.5*2.5, -1.25*0.9) {$\loc_8$};
        \node[state] (consistent_b_1) at (2.5*2.5, -2.25*0.9) {$\loc_9$};
        \node[state, final] (inconsistent) at (4.3*2.5, 0*0.9) {$\loc_{10}$};

        \path [NFApath]
            (init) edge[bend left] node[above,align=center, above left, pos=0.95] {$(\tuple{\styleact{a}, \exNIzero), \var[1]}$} (init_a_0)
            (init) edge node[above,align=center, above] {$(\tuple{\styleact{a}, \exNIone), \var[1]}$} (init_a_1)
            (init) edge node[above,align=center, below] {$(\tuple{\styleact{b}, \exNIzero), \var[1]}$} (init_b_0)
            (init) edge[bend right] node[above,align=center, below left, pos=0.95] {$(\tuple{\styleact{b}, \exNIone), \var[1]}$} (init_b_1)
            (init_a_0) edge[bend left] node[above,align=center, above right, pos=0.05] {$(\tuple{\styleact{a}, \exNIzero), \var[2]}$} (consistent)
            (init_a_1) edge node[above,align=center, above] {$(\tuple{\styleact{a}, \exNIone), \var[2]}$} (consistent)
            (init_b_0) edge node[above,align=center, below] {$(\tuple{\styleact{b}, \exNIzero), \var[2]}$} (consistent)
            (init_b_1) edge[bend right] node[above,align=center, below right, pos=0.05] {$(\tuple{\styleact{b}, \exNIone), \var[2]}$} (consistent)
            (consistent) edge[bend left=3] node[above,sloped,align=center] {$(\tuple{\styleact{a}, \exNIzero), \var[1]}$} (consistent_a_0)
            (consistent) edge[bend left=5] node[above,sloped,align=center,pos=0.7] {$(\tuple{\styleact{a}, \exNIone), \var[1]}$} (consistent_a_1)
            (consistent) edge[bend left=5] node[above,sloped,pos=0.7] {$(\tuple{\styleact{b}, \exNIzero), \var[1]}$} (consistent_b_0)
            (consistent) edge[bend left=3] node [above, sloped,,pos=0.7] {$(\tuple{\styleact{b}, \exNIone), \var[1]}$} (consistent_b_1)
            (consistent_a_0) edge[bend left=3] node[below,sloped,align=center,pos=0.3] {$(\tuple{\styleact{a}, \exNIzero), \var[2]}$} (consistent)
            (consistent_a_1) edge[bend left=5] node[below,sloped,align=center,pos=0.3] {$(\tuple{\styleact{a}, \exNIone), \var[2]}$} (consistent)
            (consistent_b_0) edge[bend left=5] node[below,sloped,pos=0.3] {$(\tuple{\styleact{b}, \exNIzero), \var[2]}$} (consistent)
            (consistent_b_1) edge[bend left=3] node[below, sloped] {$(\tuple{\styleact{b}, \exNIone), \var[2]}$} (consistent)
            (consistent_a_0) edge[bend left=30] node[above,align=center, above right,pos=0.6] {$(\tuple{\styleact{a}, \exNIone), \var[2]}$} (inconsistent)
            (consistent_a_1) edge node[below,sloped,pos=0.3] {$(\tuple{\styleact{a}, \exNIzero), \var[2]}$} (inconsistent)
            (consistent_b_0) edge node[above,sloped,pos=0.3] {$(\tuple{\styleact{b}, \exNIone), \var[2]}$} (inconsistent)
            (consistent_b_1) edge[bend right=30] node[above,align=center, below right,pos=0.6] {$(\tuple{\styleact{b}, \exNIzero), \var[2]}$} (inconsistent)
        ;
        \end{tikzpicture}
        }
        \caption{Matching interference.}%
        \label{figure:example:interference_NFA}
	\end{subfigure}
	\caption{Examples of NAAs.}\label{figure:example-NAAs}
\end{figure}
\begin{example}[interference]\label{example:interference_NAA}
 \emph{Noninterference}~\cite{Smith07} is one of the most typical examples of hyperproperties.
 A program~$P$ satisfies noninterference if for any memory states $\mu, \nu \in \mathbb{M}$ that agree on public variables,
 the memory states after running~$P$ from $\mu$ and~$\nu$ with an input sequence $\word \in \Alphabet_{I}^*$ also agree on public variables.
 Hyper pattern matching can extract witnesses of violation of noninterference.
 For instance, the NAA %
given in \cref{figure:example:interference_NFA}
 accepts evidences of interference for input actions $\Alphabet_{I} = \{a, b\}$, and public memory states $\mathbb{M} = \{\exNIzero, \exNIone\}$.
 Technically, the NAA in \cref{figure:example:interference_NFA} looks for pairs of executions starting from the same value of the variable, following the same sequence of actions, but leading to different values of the variable; for example, the execution going through $\loci{0}$, $\loci{1}$, $\loci{5}$, $\loci{6}$, $\loci{10}$ detects a pair of executions starting with variable value~$\exNIzero$, reading two ``$a$''s, but the first execution (encoded by~$\var[1]$) ends with value~$\exNIzero$ while the second one ends with value~$\exNIone$---which violates noninterference.
\end{example}

We evaluate hyper pattern matching from both theoretical and empirical perspectives.
Theoretically, we prove that it is \NP{}-complete to decide the nonemptiness of the match set, which indicates the intractability of hyper pattern matching.
Specifically, the time complexity of our algorithm is
exponential with respect to the number of directions,
polynomial with respect to the maximum length of the monitored words, and
quadratic with respect to the size of the NAA.\@
Empirically, we implement a prototype tool \ourTool{} for hyper pattern matching in the context of monitoring, and evaluate its efficiency via experiments.
We propose a naive algorithm and two heuristics to improve its efficiency by
\begin{oneenumerate}%
	\item skipping unnecessary matching trials inspired by efficient string matching algorithms, and
	\item pruning matching candidates by first performing non-hyper pattern matching over automata projected over directions.
\end{oneenumerate}
\ourTool{} can handle words with thousands of letters within one minute for \LongVersion{several }benchmarks with two directions, including the NAA in \cref{figure:example:interference_NFA},
which suggests the usefulness of hyper pattern matching for analysing a reasonably sized set of logs.

\fakeParagraph{Contributions}
Our contributions are summarized as follows:
\begin{enumerate}
	\item we propose hyper pattern matching, a generalization of pattern matching across multiple words or multiple portions of the same word;
    \item we show that nonemptiness checking of the match set in hyper pattern matching is \NP{}-complete;
	\item we provide a naive algorithm for hyper pattern matching as well as two heuristics to enhance its efficiency; %
	\item we implement our algorithms into a tool \ourTool{}, and demonstrate its capabilities over several benchmarks.
\end{enumerate}

\fakeParagraph{Related work}

NAAs, which is also called multi-tape automata~\cite{RS59}, have been studied in various domains with some variations, \eg{}~\cite{RS59,Furia12,Worrell13}. We mostly follow the formulation and terminologies in~\cite{GMO21}. %
Although the membership problem (\ie{} determining if a tuple of words is accepted by an NAA) has been studied well,
the pattern matching problem we study (\ie{} from a set of words, returning a tuple of words with intervals such that the tuple of words projected to the intervals is accepted by an NAA) has not been studied, to the best of our knowledge.
\LongVersion{Pattern matching has been extended for two-dimensional settings.}
In~\cite{Bird77}, two-dimensional pattern\LongVersion{ (actually string)} matching is considered, \ie{} matching a two-dimensional array of symbols in a text itself represented as a two-dimensional array.
In~\cite{AF92}, two-dimensional pattern matching is extended for a set of patterns, called two-dimensional dictionary matching.
Although these problems have been extensively studied~\cite{ABF94,KPR00,ZM05,EGGS25},
to the best of our knowledge, all these works focus on two-dimensional patterns without branching or loops, unlike our hyper pattern matching supporting multi-dimensional ``regular'' patterns against word sets of an arbitrary size.
Various algorithms have been proposed for monitoring hyperproperties~\cite{DFR12,AB16,BSB17,FHST18,FHST19,Hahn19,FHST20,AAAF22,CH23,AAAFGW24,BFFM24,CHdC24}. Most of these algorithms\LongVersion{ (\eg{}~\cite{AB16,BSB17,FHST18,FHST19,Hahn19,FHST20})} use HyperLTL~\cite{CFKMRS14} to represent the monitored property.
Due to the synchronous nature of HyperLTL, these algorithms cannot handle asynchronous hyperproperties, such as stuttering robustness (see \cref{example:robustness}).
The same limitation also applies to the algorithms using other related logics, such as Hyper-$\mu$HML used in~\cite{AAAF22,AAAFGW24}.
Moreover, when a violation of the monitored property is detected, these algorithms only return a Boolean verdict, a (minimal) subset of complete traces, or relevant prefixes of traces.
In contrast, our algorithm identifies the tuples of subwords matching the given property, which are finer-grained witnesses.

Recently, several papers have proposed monitoring algorithms for asynchronous hyperproperties.
In~\cite{CH23}, an automata-based formalism\LongVersion{ (``multi-trace prefix transducers'')} is introduced for monitoring asynchronous hyperproperties.
\emph{Hypernode automata}~\cite{BHNC23} is another automata-based formalism for asynchronous hyperproperties.
Each state of a hypernode automaton is labeled with a relational constraint on words represented by \emph{hypernode logic}.
\emph{Extended hypernode logic}~\cite{CHdC24} is an extension of hypernode logic with regular expressions and the stutter-reduction operation to reason about
\begin{oneenumerate}%
 \item the structure of the words and
 \item the synchronous and asynchronous comparison between the words.
\end{oneenumerate}
In~\cite{BFFM24}, a monitoring algorithm for $\mathrm{Hyper}^2\mathrm{LTL}_f$, a temporal logic representing second-order hyperproperties, is proposed.
The witnesses provided by these algorithms are also less informative than ours, analogous to the monitoring algorithms for synchronous hyperproperties.
Nevertheless, an extension of our algorithm to a more general class of hyperproperties, such as second-order hyperproperties, is one of the future directions.

In addition to the formalisms above, various logics have been proposed for representing asynchronous hyperproperties, such as A-HyperLTL~\cite{BCBFS21}, $\mathrm{HyperLTL}_S$~\cite{BPS21}, $\mathrm{HyperLTL}_C$~\cite{BPS21}, $\mathrm{GHyperLTL}_{S + C}$~\cite{BBST24}, HyperMTL~\cite{BPS20}, and $H_\mu$~\cite{GMO21}.
In~\cite{GMO21}, a construction of alternating asynchronous parity automata from $H_\mu$ formulas is shown.
A similar construction of NAAs to support these logics in hyper pattern matching is a future work.
\section{Preliminaries}\label{section:preliminaries}

We write $\setN$ and $\setNplus$ for the naturals and positive naturals.
For a partial function $f\colon Y \partfun Z$, we denote its domain by~$\domain(f)$.
For a set~$Y$, we denote by~$\powerset{Y}$ the powerset of~$Y$.
For a set~$Y$, we denote its cardinality by~$\cardinality{Y}$.
An \emph{alphabet} is a nonempty finite set~$\Alphabet$ of letters.
A (finite) \emph{word} over~$\Alphabet$ is a finite sequence of letters from~$\Alphabet$.
The \emph{empty word} is denoted by~$\emptyword$, and the set of all finite words is denoted by~$\Alphabet^*$.
For a word $\word = \action_1 \action_2 \cdots \action_n$,
we use $\slice{\word}{i}{j}$ to denote the subword $\action_{i} \action_{i+1} \dots \action_{j}$.
We write the $i$-th letter of a word $\word \in \Alphabet^*$ as $\wordi{i}$.
A \emph{language} is a subset of~$\Alphabet^*$.

\begin{definition}[NFA]
	A \emph{nondeterministic finite-word automaton} (NFA) is a tuple $\NFA = \tuple{\Alphabet, \Loc, \InitLoc, \Edges, \FinalLoc}$,
		where~$\Alphabet$ is an alphabet,
		$\Loc$ is a \LongVersion{nonempty }finite set of states,
		$\InitLoc \subseteq \Loc$ is a set of initial states,
		$\FinalLoc \subseteq \Loc$ is a set of accepting states,
		and $\Edges \subseteq \Loc \times \Alphabet \times \Loc$ %
 is a transition relation.
\end{definition}

A \emph{deterministic finite-word automaton (DFA)} is an NFA such that $\InitLoc$ is a singleton and for any $\loc \in \Loc$ and $\sigma \in \Alphabet$, there is exactly one $\loc' \in \Loc$ satisfying $(\loc, \sigma, \loc') \in \Edges$.
For DFAs, we regard $\Edges$ as a transition \emph{function}.
Given a word $\word = \sigma_1 \sigma_2 \cdots \sigma_n$ over~$\Alphabet$, a \emph{run} of~$\NFA$ on~$\word$ is a sequence of states $(\loc_0, \loc_1, \cdots, \loc_n)$ such that $\loc_0 \in \InitLoc$ and, for every $0 < i \leq n$, it holds that $(\loc_{i-1}, \action_i, \loc_i) \in \Edges$.
The run is \emph{accepting} if $\loc_n \in \FinalLoc$.
We say that $\NFA$ \emph{accepts}~$\word$ if there exists an accepting run of~$\NFA$ on~$\word$.
The \emph{language} $\Lg(\NFA)$ of~$\NFA$ is the set of all words accepted by~$\NFA$.
We use \emph{nondeterministic asynchronous finite automata (NAAs)}~\cite{GMO21} to represent asynchronous hyperproperties.
Intuitively, an NAA is an NFA equipped with \emph{directions} to asynchronously read multiple words.

\begin{definition}[NAA]
 A \emph{nondeterministic asynchronous finite automaton} (NAA) is a tuple
 $\NAA = \tuple{\Alphabet, \Vars, \Loc, \InitLoc, \Edges, \FinalLoc}$,
 where~$\Alphabet$, $\Loc$, $\InitLoc$, and $\FinalLoc$ are the same as in an NFA,
 and $\Vars = \{\var[1], \var[2], \dots, \var[k]\}$ is a set of \emph{directions}, and
 $\Edges \subseteq \Loc \times \Alphabet \times \Vars \times \Loc$ is a transition relation.
\end{definition}

For an NAA $\NAA = \tuple{\Alphabet, \Vars, \Loc, \InitLoc, \Edges, \FinalLoc}$,
we let the underlying NFA as $\asNFA{\NAA} = \tuple{\Alphabet \times \Vars, \Loc, \InitLoc, \Edges, \FinalLoc}$, \ie{}
we use $\Alphabet \times \Vars$ as the alphabet and deem $\Edges$ as a transition relation of an NFA.\@
For a word~$\extendedword$ over $\Alphabet \times \Vars$ and $\var \in \Vars$,
we let $\project{\extendedword}{\var} \in \Alphabet^*$ be the word constructed by
\begin{ienumerate}%
 \item removing the letters $\tuple{a, {\var}'}$ with ${\var}' \neq \var$ and
 \item projecting to the first element of each letter.%
\end{ienumerate}
We naturally extend $\projectsymbol$ to languages, \ie{} for $\Lg \subseteq (\Alphabet \times \Vars)^*$, $\project{\Lg}{\var} = \{\project{\extendedword}{\var} \mid \extendedword \in \Lg\}$.
An NAA $\NAA$ accepts $k$-tuple $\tuple{\word[1], \word[2], \dots, \word[k]}$ of words if
there is $\extendedword \in \Lg(\asNFA{\NAA})$ satisfying $\word[\var] = \project{\extendedword}{\var}$ for each $\var \in \Vars$.
We let $\hyperLg(\NAA)$ be the set of $k$-tuples of words accepted by $\NAA$.

\begin{example}\label{example:NAA}
    Let us revisit \cref{example:intro:counting} in a more formal manner.
	Consider the NAA $\NAA = \tuple{\Alphabet, \Vars, \Loc, \InitLoc, \Edges, \FinalLoc}$,
	with directions $\Vars = \{\var[1], \var[2]\}$.
    \cref{figure:example} illustrates~$\NAA$.
    Since we have $\Lg(\NFA) = \tuple{\startend, 1} \tuple{\startend, 1} (\tuple{a, 1} \tuple{b, 2})^* \tuple{\startend, 1}$,
    a $2$-tuple $\tuple{\word[1], \word[2]}$ of words is accepted by~$\NAA$ if and only if we have 
    $\word[1] = \startend a^n \startend$ and $\word[2] = \startend b^n $ for some $n \in \setN$,
    \ie{} $\NAA$ accepts a pair of words with the same number of ``$a$''s (preceded and followed by a~``$\startend$'') and of~``$b$''s (preceded by a~``$\startend$'').
\end{example}
 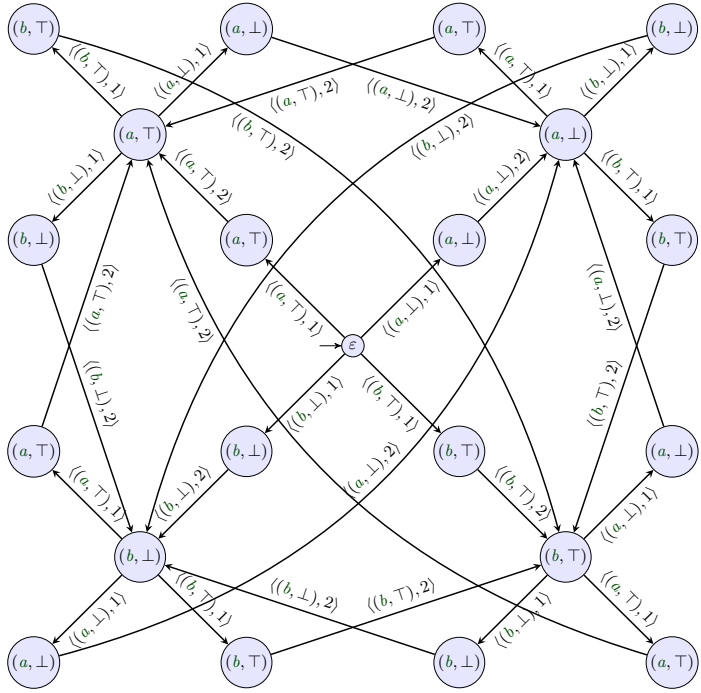
\begin{figure}[t]
 \centering
 \scalebox{0.70}{\begin{tikzpicture}[NFA]
  \node[state,initial] (init) at (0, 0) {$\emptyword$};
  \node[state] (a_top) at (-2.0, 2.0) {$(\styleact{a}, \top)$};
  \node[state] (a_bot) at (2.0, 2.0) {$(\styleact{a}, \bot)$};
  \node[state] (b_top) at (2.0, -2.0) {$(\styleact{b}, \top)$};
  \node[state] (b_bot) at (-2.0, -2.0) {$(\styleact{b}, \bot)$};
  \node[state] (a_top_2) at (-4.0, 4.0) {$(\styleact{a}, \top)$};
  \node[state] (a_bot_2) at (4.0, 4.0) {$(\styleact{a}, \bot)$};
  \node[state] (b_top_2) at (4.0, -4.0) {$(\styleact{b}, \top)$};
  \node[state] (b_bot_2) at (-4.0, -4.0) {$(\styleact{b}, \bot)$};
  \node[state] (a_bot_a_top) at (-2.0, 6.0) {$(\styleact{a}, \bot)$};
  \node[state] (b_top_a_top) at (-6.0, 6.0) {$(\styleact{b}, \top)$};
  \node[state] (b_bot_a_top) at (-6.0, 2.0) {$(\styleact{b}, \bot)$};
  \node[state] (a_top_a_bot) at (2.0, 6.0) {$(\styleact{a}, \top)$};
  \node[state] (b_top_a_bot) at (6.0, 2.0) {$(\styleact{b}, \top)$};
  \node[state] (b_bot_a_bot) at (6.0, 6.0) {$(\styleact{b}, \bot)$};
  \node[state] (a_top_b_top) at (6.0, -6.0) {$(\styleact{a}, \top)$};
  \node[state] (a_bot_b_top) at (6.0, -2.0) {$(\styleact{a}, \bot)$};
  \node[state] (b_bot_b_top) at (2.0, -6.0) {$(\styleact{b}, \bot)$};
  \node[state] (a_top_b_bot) at (-6.0, -2.0) {$(\styleact{a}, \top)$};
  \node[state] (a_bot_b_bot) at (-6.0, -6.0) {$(\styleact{a}, \bot)$};
  \node[state] (b_top_b_bot) at (-2.0, -6.0) {$(\styleact{b}, \top)$};

  \path [NFApath]
  (init) edge[bend left=0] node[below,sloped] {$\tuple{(\styleact{a}, \top), \var[1]}$} (a_top)
  (init) edge[bend left=0] node[below,sloped] {$\tuple{(\styleact{a}, \bot), \var[1]}$} (a_bot)
  (init) edge[bend left=0] node[below,sloped] {$\tuple{(\styleact{b}, \top), \var[1]}$} (b_top)
  (init) edge[bend left=0] node[below,sloped] {$\tuple{(\styleact{b}, \bot), \var[1]}$} (b_bot)
  (a_top) edge[bend left=0] node[above,sloped] {$\tuple{(\styleact{a}, \top), \var[2]}$} (a_top_2)
  (a_bot) edge[bend left=0] node[above,sloped] {$\tuple{(\styleact{a}, \bot), \var[2]}$} (a_bot_2)
  (b_top) edge[bend left=0] node[above,sloped] {$\tuple{(\styleact{b}, \top), \var[2]}$} (b_top_2)
  (b_bot) edge[bend left=0] node[above,sloped] {$\tuple{(\styleact{b}, \bot), \var[2]}$} (b_bot_2)

  (a_top_2) edge[bend left=0] node[above,sloped] {$\tuple{(\styleact{a}, \bot), \var[1]}$} (a_bot_a_top)
  (a_top_2) edge[bend left=0] node[above,sloped] {$\tuple{(\styleact{b}, \top), \var[1]}$} (b_top_a_top)
  (a_top_2) edge[bend left=0] node[above,sloped] {$\tuple{(\styleact{b}, \bot), \var[1]}$} (b_bot_a_top)
  (a_bot_2) edge[bend left=0] node[above,sloped] {$\tuple{(\styleact{a}, \top), \var[1]}$} (a_top_a_bot)
  (a_bot_2) edge[bend left=0] node[above,sloped] {$\tuple{(\styleact{b}, \top), \var[1]}$} (b_top_a_bot)
  (a_bot_2) edge[bend left=0] node[above,sloped] {$\tuple{(\styleact{b}, \bot), \var[1]}$} (b_bot_a_bot)
  (b_top_2) edge[bend left=0] node[above,sloped] {$\tuple{(\styleact{a}, \top), \var[1]}$} (a_top_b_top)
  (b_top_2) edge[bend left=0] node[below,sloped] {$\tuple{(\styleact{a}, \bot), \var[1]}$} (a_bot_b_top)
  (b_top_2) edge[bend left=0] node[below,sloped] {$\tuple{(\styleact{b}, \bot), \var[1]}$} (b_bot_b_top)
  (b_bot_2) edge[bend left=0] node[above,sloped] {$\tuple{(\styleact{a}, \top), \var[1]}$} (a_top_b_bot)
  (b_bot_2) edge[bend left=0] node[below,sloped] {$\tuple{(\styleact{a}, \bot), \var[1]}$} (a_bot_b_bot)
  (b_bot_2) edge[bend left=0] node[above,sloped] {$\tuple{(\styleact{b}, \top), \var[1]}$} (b_top_b_bot)

  (a_bot_a_top) edge[bend left=0] node[below,sloped] {$\tuple{(\styleact{a}, \bot), \var[2]}$} (a_bot_2)
  (b_top_a_top) edge[bend left=30] node[below,sloped,pos=0.3] {$\tuple{(\styleact{b}, \top), \var[2]}$} (b_top_2)
  (b_bot_a_top) edge[bend right=0] node[above,sloped] {$\tuple{(\styleact{b}, \bot), \var[2]}$} (b_bot_2)
  (a_top_a_bot) edge[bend right=0] node[below,sloped] {$\tuple{(\styleact{a}, \top), \var[2]}$} (a_top_2)
  (b_top_a_bot) edge[bend left=0] node[above,sloped] {$\tuple{(\styleact{b}, \top), \var[2]}$} (b_top_2)
  (b_bot_a_bot) edge[bend right=30] node[below,sloped,pos=0.3] {$\tuple{(\styleact{b}, \bot), \var[2]}$} (b_bot_2)
  (a_top_b_top) edge[bend left=28] node[below,sloped, pos=0.8] {$\tuple{(\styleact{a}, \top), \var[2]}$} (a_top_2)
  (a_bot_b_top) edge[bend right=0] node[below,sloped] {$\tuple{(\styleact{a}, \bot), \var[2]}$} (a_bot_2)
  (b_bot_b_top) edge[bend left=0] node[above,sloped] {$\tuple{(\styleact{b}, \bot), \var[2]}$} (b_bot_2)
  (a_top_b_bot) edge[bend left=0] node[below,sloped] {$\tuple{(\styleact{a}, \top), \var[2]}$} (a_top_2)
  (a_bot_b_bot) edge[bend right=30] node[above,sloped] {$\tuple{(\styleact{a}, \bot), \var[2]}$} (a_bot_2)
  (b_top_b_bot) edge[bend left=0] node[above,sloped] {$\tuple{(\styleact{b}, \top), \var[2]}$} (b_top_2)
 ;
 \end{tikzpicture}}
 \caption{Example: matching evidences of violations of stuttering robustness, with $\Alphabet = \Alphabet_{I} \times \Alphabet_{O}$, $\Alphabet_{I} = \{\styleact{a}, \styleact{b}\}$, and $\Alphabet_{O} = \{\top, \bot\}$. The self-loops, the accepting state, and the transitions to the accepting state are omitted: for each state labeled with $x \in \Alphabet$, we have a self-loop labeled with $(x, \var)$, where $\var \in \Vars$ is the direction of the incoming transition; for each state labeled with $(\sigma, \gamma)$, with an outgoing transition labeled with $\tuple{(\sigma, \gamma), 2}$ and with an incoming transition from a state labeled with $\emptyword$ or $(\sigma', \gamma)$, where $\sigma \neq \sigma'$, we have a transition to the accepting state labeled with $\tuple{(\sigma, \neg \gamma), 2}$, where $\neg \top = \bot$ and $\neg \bot = \top$.}%
  \label{figure:example:robustness_NAA}
 \end{figure}
\begin{example}[stuttering robustness]%
 \label{example:robustness}%
 \emph{Robustness} is another common hyperproperty.
 Robustness requires that for two similar inputs, the system's behavior must be the same (or similar).
 One of its instances is robustness with respect to stuttering, \ie{}
 if two sequences $\word, {\word}' \in \Alphabet_{I}^*$ of inputs are identical by reducing stuttering,
 these sequences accompanied by their corresponding outputs, (\ie{} $\tilde{\word}, \tilde{\word}' \in (\Alphabet_{I} \times \Alphabet_{O})^*$ whose projections to $\Alphabet_{I}^*$ are $\word$ and ${\word}'$, respectively) must be also the same after removing the stuttering.
 The NAA $\NAA = \tuple{\Alphabet, \Vars, \Loc, \InitLoc, \Edges, \FinalLoc}$,
 with directions $\Vars = \{\var[1], \var[2]\}$
 in \cref{figure:example:robustness_NAA}, %
 accepts evidences of non-robust execution.
 Intuitively, after reading each letter from direction~1, it asserts that the next letter from direction~2 has the same output if its input is the same as direction~1, where the stuttering is reduced by the self-loops, which are omitted in \cref{figure:example:robustness_NAA}.
\end{example}

Recall that we also showed in \cref{example:interference_NAA} that \mbox{(non-)}interference can be encoded using NAAs.
In addition, we give in \cref{example:network} in \cref{appendix:example:network} an\LongVersion{ additional} example of monitoring packets of similar size over a network (typically using a UDP-based protocol, \eg{} RTP~\cite{rfc1889}).

\section{Hyper Pattern Matching Problem}\label{section:problem}

\LongVersion{%
The formal definition of hyper pattern matching is as follows.
}

\defProblem{Hyper pattern matching}{
A finite set $\wordSet \in \powerset{\Alphabet^*}$ of words and
an NAA $\NAA = \tuple{\Alphabet, \Vars, \Loc, \InitLoc, \Edges, \FinalLoc}$ with $\Vars = \{\var[1], \dots, \var[k]\}$}{
The match set $\Match(\NAA, \wordSet) = \big\{ \tuple{(\word[1], i^{1}, j^1), \dots, (\word[k], i^k, j^k)} \in {(\wordSet \times \setN \times \setN)}^k \mid \tuple{\slice{\word[1]}{i^1}{j^1}, \slice{\word[2]}{i^2}{j^2}, \dots, \slice{\word[k]}{i^k}{j^k}} \in \hyperLg(\NAA) \big\}$
}
\begin{example}\label{example:matchset}
	Consider again the NAA~$\NAA$ in \cref{figure:example}. %
	Let $\wordSet = \{ \word \}$ be a singleton word set, with $\word = d\startend{}aa\startend{}bbb\startend{}aaa\startend{}\startend{}e$.
	For instance, $\tuple{(\word, 2, 5), (\word, 5, 7)} \in \Match(\NAA, \wordSet)$ holds because $\extendedword = \tuple{\startend{}, 1} \tuple{\startend{}, 2} \tuple{a, 1} \tuple{b, 2} \tuple{a, 1} \tuple{b, 2} \tuple{\startend{}, 1}$ satisfies $\extendedword \in \Lg(\asNFA{\NAA})$, $\project{\extendedword}{1} = \slice{\word}{2}{5}$, and $\project{\extendedword}{2} = \slice{\word}{5}{7}$.
	The match set $\Match(\NAA, \wordSet)$ is
	\[\Big\{ \tuple{ (\word, 2, 5 ) , (\word, 5, 7 )} , \tuple{ (\word, 9, 13 ) , (\word, 5, 8 )} \Big\} \cup \Big\{ \tuple{ (\word, 13, 14 ) , (\word, i, i)} \mid i \in \{ 2,5,9,13,14 \} \Big\} \text{.}\]
\end{example}

Deciding the nonemptiness of the match set is \NP{}-complete. This complexity suggests that the exponential blowup in the worst case of the hyper pattern matching algorithms we propose later in \cref{section:algorithm,section:heuristics} is inevitable.
(The proofs of this result and subsequent results are in the appendix.)

\begin{restatable}{theorem}{theoremMatchSetNPcomplete}
 \label{theorem:emptiness_matching_np_complete}
 The nonemptiness decision problem for the match set $\Match(\NAA, \wordSet)$ for an NAA $\NAA$ and a finite set $\wordSet$ of words is \NP{}-complete.
\end{restatable}
\section{A naive algorithm for hyper pattern matching}\label{section:algorithm}

Before presenting a naive algorithm for hyper pattern matching, we define an auxiliary notation.
For an NFA $\NFA = \tuple{\Alphabet \times \Vars, \Loc, \InitLoc, \Edges, \FinalLoc}$ with $\Vars = \{\var[1], \var[2], \dots, \var[k]\}$,
we define a relation ${\to} \subseteq \big({(\Alphabet^*)}^k \times \Loc \big) \times \big({(\Alphabet^*)}^k \times \Loc\big)$ such that
$\tuple{\worda[1], \worda[2], \dots, \worda[k], \loc} \to \tuple{\wordb[1], \wordb[2], \dots, \wordb[k], \loc'}$
if and only if
there is $l \in \{1, 2, \dots, k\}$ and $(\loc, (\action, \var[l]), \loc') \in \Edges$ %
satisfying $\worda[l] = \action \cdot \wordb[l]$, and for any $m \neq l$, $\worda[m] = \wordb[m]$ holds.
\begin{algorithm}[tb]
    \caption{A naive algorithm $\algoHPM$ for hyper pattern matching.}\label{algorithm:naive}

    \footnotesize

    \KwIn{A finite set $\wordSet \subseteq \Alphabet^*$ of words and an NAA $\NAA = \tuple{\Alphabet, \Vars, \Loc, \InitLoc, \Edges, \FinalLoc}$ with $\Vars = \{\var[1], \var[2], \dots, \var[k]\}$
 }
    \KwOut{The match set $\Match(\NAA, \wordSet) = \big\{\tuple{(\word[1], i^1, j^1), \dots, (\word[k], i^k, j^k)} \mid \tuple{\slice{\word[1]}{i^1}{j^1}, \slice{\word[2]}{i^2}{j^2}, \dots, \slice{\word[k]}{i^k}{j^k}} \in \hyperLg(\NAA) \big\}$}

    $\Match \gets \emptyset$\;

    \tcp{Priority queue of the beginning of the matching trials}
    $\waitingQueue \gets \big\{\tuple{i_1, \dots,i_k, \word[1], \dots, \word[k]} \mid \forall m \in \{1,\dots,k\}.\, \word[m] \in \wordSet, 1 \leq i_m \leq \cardinality{\word[m]} \big\}$\nllabel{algorithm:naive:initialize_waiting_queue}\;

    \While{$\cardinality{\waitingQueue} > 0$} {
        \tcp{Pop the ``smallest'' element from the priority queue}
        \KwPop{} $\tuple{i_1, \dots,i_k, \word[1], \dots, \word[k]}$ \KwFrom{} $\waitingQueue$\nllabel{algorithm:naive:pop} \;

        $p_1, p_2, \dots, p_k  \gets i_1, i_2, \dots, i_k$\nllabel{algorithm:naive:setp}

        $\currentConfigurations \gets \big\{\tuple{\wordi[1]{i_1}, \wordi[2]{i_2}, \dots, \wordi[k]{i_k}, \initLoc} \mid \initLoc \in \InitLoc \big\}$\nllabel{algorithm:naive:setCurrentC}
		\tcp*[f]{Start new matching trials}

        \While{$\currentConfigurations \neq \emptyset \land \exists m \in \{1, 2, \dots, k\}.\, p_m < |\word[m]|$ \nllabel{algorithm:naive:while:condition}} {
            \For{$m \in \{1, 2, \dots, k\}$ satisfying $p_m < |\word[m]|$ \nllabel{algorithm:naive:for}} {
                $p_m \gets p_m + 1$ \nllabel{algorithm:naive:setpm+1}
                \tcp{Read the $(p_m + 1)$-th letter ${\wordi[m]{p_m + 1}}$ of $\word[m]$}
                $\currentConfigurations \gets \big\{ \assign{\tuple{\worda[1], \dots, \worda[k], \loc}}{\worda[m]}{\worda[m] \cdot {\wordi[m]{p_m}}} \mid \tuple{\worda[1], \dots, \worda[k], \loc} \in \currentConfigurations \big\}$\nllabel{algorithm:naive:updateC}\;

                $\currentConfigurations' \gets \currentConfigurations$
                \tcp{Compute the next configuration}

                \While{$\currentConfigurations' \neq \emptyset$\nllabel{algorithm:naive:transitions:begin}} {
                    \KwPop{} $\tuple{\worda[1], \dots, \worda[k], \loc}$ \KwFrom{} $\currentConfigurations'$ \;
                    \For{$\tuple{\wordb[1], \dots, \wordb[k], \loc'} \not\in\currentConfigurations$ \st{} $\tuple{\worda[1], \dots, \worda[k], \loc} \to \tuple{\wordb[1], \dots, \wordb[k], \loc'}$} {
                        \tcp{Apply transitions}
                        $\currentConfigurations \gets \currentConfigurations \cup \{\tuple{\wordb[1], \dots, \wordb[k], \loc'}\}$;\,
                        $\currentConfigurations' \gets \currentConfigurations' \cup \{\tuple{\wordb[1], \dots, \wordb[k], \loc'}\}$\nllabel{algorithm:naive:transitions}
                    }
                }

                \For{$\tuple{\worda[1], \dots, \worda[k], \loc} \in \currentConfigurations$}{
                    \If(\tcp*[f]{Detect matching and update $\Match$}) {$\loc \in \FinalLoc$ \nllabel{algorithm:naive:detect}} {
                        \KwPush{} $\tuple{(\word[1], i_1, p_1 - |v_i|), \dots, (\word[k], i_k, p_k - |v_k|)}$ \KwTo{} $\Match$\nllabel{algorithm:naive:match}\;
                    }
                }

                \tcp{Remove the ``non-waiting'' configurations}
                $\currentConfigurations \gets \big\{ \tuple{\worda[1], \worda[2], \dots, \worda[k], \loc} \in \currentConfigurations \mid \exists m \in \{1,2,\dots,k\}.\, \worda[m] = \emptyword \big\}$\nllabel{algorithm:naive:remove}\label{algorithm:naive:transitions:end}
            }
         }
    }
    \KwRet{$\Match$}\;
\end{algorithm}
\cref{algorithm:naive} shows a naive algorithm for hyper pattern matching.
In \cref{algorithm:naive}, we use a priority queue~$\waitingQueue$ containing the information of the upcoming matching trials.
In~$\waitingQueue$, we use the lexicographic order, assuming that the set~$\wordSet$ of the examined words is totally ordered.
The exponential blowup with respect to~$k$ at \cref{algorithm:naive:initialize_waiting_queue} is most likely inevitable because the nonemptiness checking of $\Match(\NAA, \wordSet)$ is already \NP{}-hard (\cref{theorem:emptiness_matching_np_complete})\footnote{Here, we show an algorithm fully constructing $\waitingQueue$ at the beginning, to simplify the explanation of skipping in~\cref{ss:skipping}. In our tool \ourTool{}, we lazily construct $\waitingQueue$ and memorize the skipped indices separately to reduce the memory usage.}.
One can easily enforce additional constraints (\eg{} one word can be used in one matching only once) to the match set by modifying the definition of $\waitingQueue$.

For each $\tuple{i_1, \dots,i_k, \word[1], \dots, \word[k]} \in \waitingQueue$, we try to find matching trials starting from $\wordi[1]{i_1}, \dots \wordi[k]{i_k}$ (\crefrange{algorithm:naive:while:condition}{algorithm:naive:transitions:end}).
In matching trials, we maintain the set $\currentConfigurations$ of configurations.
Each configuration $\tuple{\worda[1], \dots, \worda[k], \loc}$ consists of a tuple $\tuple{\worda[1], \dots, \worda[k]}$ of words that are read by \cref{algorithm:naive} but not yet fed to the NAA $\NAA$ and the current state $\loc$.
We update $\currentConfigurations$ by appending a new letter (\cref{algorithm:naive:updateC}), applying transitions (\crefrange{algorithm:naive:transitions:begin}{algorithm:naive:transitions}), and removing ``non-waiting'' configurations, \ie{} the configurations $\tuple{\worda[1], \dots, \worda[k], \loc}$ with $\worda[l] \neq \emptyword$ for any $l \in \Vars$ (\cref{algorithm:naive:remove}).
These configurations can be removed because no additional transitions are enabled by appending letters.

\begin{example}\label{example:algorithm:naive}
	Consider again the NAA~$\NAA$ in \cref{example:NAA}.
	Let $\wordSet = \{ \word \}$ be a singleton word set, with $\word = \startend{}a\startend{}b$.
	The initial priority queue (with an exponential blowup) at \cref{algorithm:naive:initialize_waiting_queue} is
	\[ \begin{array}{l l}
	\waitingQueue = \big\{ &
	          \tuple{1, 1, \word, \word},
          \tuple{1, 2, \word, \word},
          \tuple{1, 3, \word, \word},
          \tuple{1, 4, \word, \word},
          \tuple{2, 1, \word, \word},
          \tuple{2, 2, \word, \word},
        \\
        &
          \tuple{2, 3, \word, \word},
          \tuple{2, 4, \word, \word},
          \tuple{3, 1, \word, \word},
          \tuple{3, 2, \word, \word},
          \tuple{3, 3, \word, \word},
          \tuple{3, 4, \word, \word},
          \\
          & \tuple{4, 1, \word, \word},
          \tuple{4, 2, \word, \word},
          \tuple{4, 3, \word, \word},
          \tuple{4, 4, \word, \word}
		\big\}.
	\end{array}\]
	Then, at \cref{algorithm:naive:pop}, we pop from~$\waitingQueue$ the smallest element, \ie{} $\tuple{1, 1, \word, \word}$.
	We let $p_1, p_2 \gets 1, 1$ (\cref{algorithm:naive:setp}).
	We let $\currentConfigurations \gets \{\tuple{\word_1, \word_1, \loci{0} } \} = \{\tuple{\startend{}, \startend{}, \loci{0} } \}$ (\cref{algorithm:naive:setCurrentC}).
	Because $\currentConfigurations \neq \emptyset$ and both $m = 1$ and $m = 2$ satisfy $p_m = 1 < |\word| = 4$ (\cref{algorithm:naive:while:condition}), we enter the \inlineAlgoStyleKW{while} loop.
	We iterate over both values for~$m$ (\cref{algorithm:naive:for}).
	Let us first consider $m = 1$.
	We set $p_1 \gets 2$ (\cref{algorithm:naive:setpm+1}).
	We update $\currentConfigurations \gets \{ \tuple{\startend{}a, \startend{}, \loci{0} } \}$ (\cref{algorithm:naive:updateC}).
    After applying transitions (\cref{algorithm:naive:transitions}) gives $\currentConfigurations \gets \{\tuple{\startend{}a, \startend{}, \loci{0} } , \tuple{a, \startend{}, \loci{1} }, \tuple{a, \emptyword, \loci{2} }, \tuple{\emptyword, \emptyword, \loci{3} } \}$,
	since $\tuple{\startend{}a, \startend{}, \loci{0} } \to \tuple{a, \startend{}, \loci{1} } \to \tuple{a, \emptyword, \loci{2} } \to \tuple{\emptyword, \emptyword, \loci{3} }$.
	No final configuration is reached (\cref{algorithm:naive:detect}), and therefore, no match is detected, and the non-waiting configurations (\cref{algorithm:naive:remove}) are removed, giving
	$\currentConfigurations \gets \{ \tuple{a, \emptyword, \loci{2} }, \tuple{\emptyword, \emptyword, \loci{3} } \}$.
	We then move to $m = 2$ and set $p_2 \gets 2$ (\cref{algorithm:naive:setpm+1}).
	We update $\currentConfigurations \gets \{ \tuple{a, a, \loci{2} }, \tuple{\emptyword, a, \loci{3} } \}$ (\cref{algorithm:naive:updateC}).
	No transition can be taken, and we exit the \inlineAlgoStyleKW{for} loop (\cref{algorithm:naive:for}) with $\currentConfigurations = \{\tuple{\emptyword, a, \loci{3} }\}$.
	In the second iteration of the \inlineAlgoStyleKW{while} loop, we set $p_1 \gets 3$ (\cref{algorithm:naive:setpm+1}), giving $\currentConfigurations \gets \{ \tuple{\startend{}, a, \loci{3} } \}$ (\cref{algorithm:naive:updateC}).
	No transition can be applied, and the non-waiting configurations are removed, yielding $\currentConfigurations \gets \emptyset$.
	This concludes the search for a match starting from $\tuple{1, 1, \word, \word}$ with a failure.

	Now, let us pop $\tuple{1, 3, \word, \word}$ from~$\waitingQueue$; we set $\currentConfigurations \gets \{\tuple{\startend{}, \startend{}, \loci{0} } \}$.
	We let $p_1, p_2 \gets 1, 3$ (\cref{algorithm:naive:setp}).
	In the first iteration of the \inlineAlgoStyleKW{while} loop (\cref{algorithm:naive:while:condition}), we first set $p_1 \gets 2$ (\cref{algorithm:naive:setpm+1}), and we set $\currentConfigurations \gets \{ \tuple{\startend{}a, \startend{}, \loci{0} } \}$.
	We then apply transitions. %
    After removing non-waiting configurations, we have $\currentConfigurations \gets \{ \tuple{\emptyword, \emptyword, \loci{3} } \}$.
	We then set $p_2 \gets 4$ (\cref{algorithm:naive:setpm+1}), and we set $\currentConfigurations \gets \{ \tuple{\emptyword, b, \loci{3} } \}$.
	This time, we can apply transitions, yielding $\currentConfigurations \gets \{ \tuple{\emptyword, b, \loci{3} } , \tuple{ \emptyword, \emptyword, \loci{2} } \}$.
	In the second iteration of the \inlineAlgoStyleKW{while} loop (\cref{algorithm:naive:while:condition}), we first set $p_1 \gets 3$ (\cref{algorithm:naive:setpm+1}), and we set $\currentConfigurations \gets \{ \tuple{\startend{}, b, \loci{3} } , \tuple{ \startend{}, \emptyword, \loci{2} } \}$.
	We then apply transitions, giving $\currentConfigurations \gets \{ \tuple{\startend{}, b, \loci{3} } , \tuple{ \startend{}, \emptyword, \loci{2} } , \tuple{\emptyword, \emptyword, \locfinal } \}$.
	We found an accepting state (\cref{algorithm:naive:detect}), and we update $\Match \gets \{ \tuple{ (\word, 1, 3) , (\word, 3, 4) } \}$ (\cref{algorithm:naive:match}).
	After removing non-waiting configurations, we have $\currentConfigurations \gets \{ \tuple{ \startend{}, \emptyword, \loci{2} } , \tuple{\emptyword, \emptyword, \locfinal } \}$.
	We then do not consider~$p_2$, as $p_2 < 4$ does not hold anymore (\cref{algorithm:naive:for}).
	In the third iteration of the \inlineAlgoStyleKW{while} loop (\cref{algorithm:naive:while:condition}), we set $p_1 \gets 4$ (\cref{algorithm:naive:setpm+1}), and we set $\currentConfigurations \gets \{ \tuple{ \startend{}b, \emptyword, \loci{2} } , \tuple{b, \emptyword, \locfinal } \}$.
	We then apply transitions, giving $\currentConfigurations \gets \{ \tuple{ \startend{}b, \emptyword, \loci{2} } %
 , \tuple{ b, \emptyword, \locfinal }\}$.
	We found another match $\{ \tuple{ (\word, 1, 4 - 1) , (\word, 3, 4) } \}$, which does not modify~$\Match$ as it was found before.
	Any other starting configuration will result in failure, and the final match set is---as expected---$\Match(\NAA, \wordSet) = \{ \tuple{ (\word, 1, 3) , (\word, 3, 4) } \}$.
\end{example}
\fakeParagraph{Complexity analysis}
The initial size of the priority queue $\waitingQueue$ at \cref{algorithm:naive:initialize_waiting_queue} is bounded by $\cardinality{\wordSet}^k \times (\max_{\word \in \wordSet}\cardinality{\word})^k$.
The number of iterations of the \inlineAlgoStyleKW{while} loop from \cref{algorithm:naive:while:condition} is bounded by $\max_{\word \in \wordSet}\cardinality{\word}$.
For each such iteration, the number of iterations of the \inlineAlgoStyleKW{while} loop from \cref{algorithm:naive:transitions:begin}, is bounded by $\cardinality{\Loc} \times \max_{\word \in \wordSet}^k$,
and for each iteration, at most $\cardinality{\Loc}$ configurations are added to $\currentConfigurations$ and $\currentConfigurations'$.
Overall, the time complexity of \cref{algorithm:naive} is bounded by $\complexityClass\big(\cardinality{\wordSet}^{k} \times \max_{\word \in \wordSet}\cardinality{\word}^{k+1} \times \cardinality{\Loc}^2 \big)$.

\section{Heuristics for hyper pattern matching}\label{section:heuristics}

Here, we present two heuristics to improve the efficiency of hyper pattern matching: FJS-style skipping (\cref{ss:skipping}) and projection-based pruning (\cref{ss:projection}).

\subsection{FJS-style skipping of matching trials}\label{ss:skipping}
The FJS algorithm~\cite{FJS07} is an efficient algorithm for the string matching problem: given a pattern word $\pattern$ and a target word $\target$, it finds all occurrences of $\pattern$ within $\target$.
The idea of the FJS algorithm has been used to improve the efficiency of automata-based pattern matching, \eg{}~\cite{WHS17,WAH23}.
We apply a similar idea to \cref{algorithm:naive} to improve its efficiency.

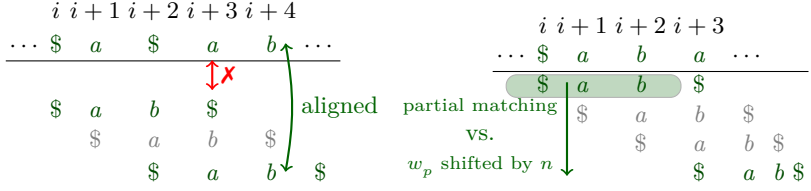
\begin{figure}[tbp]
 \begin{subfigure}[b]{0.47\textwidth}
  \centering
  \footnotesize
  \begin{tabular}{cccccccccc}
   & $i$ & $i + 1$ & $i + 2$ & $i + 3$ & $i + 4$ \\
   $\cdots$ & \styleact{\startend{}} & \styleact{a} & \styleact{\startend{}} & \tikz[remember picture,baseline=(target_qs_last.base)]{\node (target_qs_last) {\styleact{a}};} & \tikz[remember picture,baseline=(target_qs_next.base)]{\node (target_qs_next) {\styleact{b}};} & $\cdots$\\\hline
   \\
   & \styleact{\startend{}} & \styleact{a} & \styleact{b} & \tikz[remember picture,baseline=(pattern_qs_last.base)]{\node (pattern_qs_last) {\styleact{\startend{}}};} &&\\
   && \styleskippedact{\startend{}} & \styleskippedact{a} & \styleskippedact{b} & \styleskippedact{\startend{}} \\
   &&& \styleact{\startend{}} & \styleact{a} & \tikz[remember picture,baseline=(pattern_qs_next.base)]{\node (pattern_qs_next) {\styleact{b}};} & \styleact{\startend}  \\
  \end{tabular}
  \begin{tikzpicture}[remember picture, overlay]
   \path [<->,thick,red]
   (target_qs_last.south) edge[bend left=0] node[right] {\color{red}\xmark} (pattern_qs_last.north)
   ;
   \path [<->,thick,DarkGreen]
   (target_qs_next.east) edge[bend left=10] node[right] {\color{DarkGreen}aligned} (pattern_qs_next.east)
   ;
  \end{tikzpicture}
  \caption{QS-style skipping based on letter alignment.}%
 \label{figure:illustration_fjs:qs}
 \end{subfigure}
 \hfill
 \begin{subfigure}[b]{0.50\textwidth}
  \centering
  \footnotesize
  \begin{tabular}{cccccccccc}
   & $i$ & $i + 1$ & $i + 2$ & $i + 3$ \\
   $\cdots$ & \styleact{\startend{}} & \styleact{a} & \styleact{b} & \styleact{a} & $\cdots$ \\\hline
   & \styleact{\startend{}} & \tikz[remember picture,baseline=(pattern_kmp_partial_matching.base)] {\node (pattern_kmp_partial_matching){\styleact{a}}} & \styleact{b} & \styleact{\startend{}} \\
   && \styleskippedact{\startend{}} & \styleskippedact{a} & \styleskippedact{b} & \styleskippedact{\startend{}} \\
   &&& \styleskippedact{\startend{}} & \styleskippedact{a} & \styleskippedact{b} & \styleskippedact{\startend{}} \\
   &&&& \styleact{\startend{}} & \styleact{a} & \styleact{b} & \styleact{\startend{}} \\
  \end{tabular}
  \begin{tikzpicture}[remember picture, overlay]
   \draw[rounded corners,fill=DarkGreen,nearly transparent,draw opacity=0.3] ($(pattern_kmp_partial_matching.south) - (1.0,-0.05)$) rectangle ++(2.3,0.3);
   \path [->,thick,DarkGreen]
   ($(pattern_kmp_partial_matching.south) - (0.20,-0.25)$)
    edge[bend left=0] node[left,pos=0.6,align=center] {\scriptsize partial matching\\ vs.\\\scriptsize $\pattern$ shifted by $n$}
   ($(pattern_kmp_partial_matching.south) - (0.20,0.99)$)
   ;
  \end{tikzpicture}
  \caption{KMP-style skipping utilizing the latest successful partial matching.}%
 \label{figure:illustration_fjs:kmp}
 \end{subfigure}
 \caption{Illustration of skipping in the FJS algorithm for string matching, with the pattern word $\pattern = \styleact{\startend{}ab\startend{}}$. The skipped matching trials are shown in grey.}
 \label{figure:illustration_fjs}
\end{figure}
The central idea of FJS-style algorithms is to efficiently identify some matching trials as unnecessary and skip them.
FJS-style algorithms combine QS-style skipping and KMP-style skipping, which originate from the Quick Search (QS) algorithm~\cite{Sunday90} and the Knuth-Morris-Pratt (KMP) algorithm~\cite{KMP77}, respectively.

\cref{figure:illustration_fjs} illustrates the idea of skipping in the FJS algorithm for string matching.
The QS-style skipping moves the pattern word $\pattern$ so that the letter in a certain position of the target word $\target$ aligns with the same letter in $\pattern$.
We first compare the last letter of $\pattern$ (\styleact{\startend{}} in \cref{figure:illustration_fjs:qs}) with the corresponding letter in $\target$ (the $(i+3)$-th letter of $\target$, \ie{} \styleact{a}, in \cref{figure:illustration_fjs:qs}).
If they are different, we move $\pattern$ so that the letter in the target word immediately after the mismatched letter can have a matching.
In this example, we move $\pattern$ so that the $(i+4)$-th letter of $\target$, \ie{} \styleact{b}, aligns with the next occurrence of~\styleact{b} in $\pattern$.
One can efficiently perform such skipping by constructing $\SkipQS\colon \Alphabet \to \setN$ that maps the aligned letter (\styleact{b}~in \cref{figure:illustration_fjs:qs}) to the length of the skip (2~in \cref{figure:illustration_fjs:qs}) in advance.

The KMP-style skipping uses the information of the partial matching in the latest matching trial to identify unnecessary matching trials.
In the example in \cref{figure:illustration_fjs:kmp}, the latest matching trial was successful for three letters.
Based solely on this information, we know that the $i$-th to $(i+2)$-th letters of $\target$ are $\styleact{\startend{}ab}$.
Since the minimum $n \in \setNplus$ satisfying $\styleact{\startend{}ab} \cdot \Alphabet^* \cap \Alphabet^{n} \cdot \pattern \neq \emptyset$ is~3, we can skip the matching trials from the $(i + j)$-th letter of $\target$ with $j \in \{1, 2\}$.
One can efficiently perform such skipping by constructing $\SkipKMP\colon \{0,1,\dots,\cardinality{\pattern}\} \to \setN$ that maps the length of partial matching (3~in \cref{figure:illustration_fjs:kmp}) to the minimum $n$ above (3~in \cref{figure:illustration_fjs:kmp}) in advance.

In~\cite{WHS17}, an FJS-style algorithm for NFA pattern matching was proposed.
The main ideas of this extension are summarized as follows:
1)~The length $\ShortestMatching{}$ of the shortest matching is used instead of the length $\cardinality{\pattern}$ of the pattern word;
2)~The set $\LastQS{}$ of letters that can appear as the $\ShortestMatching{}$-th letter of a word accepted by the pattern NFA is constructed beforehand;
3)~A partial matching is characterized by a state of the pattern NFA instead of its length, \ie{} $\SkipKMP$ takes a state $\loc \in \Loc$ instead of $n \in  \{0,1,\dots,\cardinality{\pattern}\}$.
Here, we further generalize these ideas for hyper pattern matching by parametrising the above concepts with directions.
In this multi-directional extension, we reformulate the notion of ``skipping'' as ``invalidating some positions''.
Thanks to this formulation, we can skip matching trials focusing on \emph{each word} rather than \emph{each tuple of words}.

\begin{definition}
 [QS-style skip values]
 Let $\Vars = \{\var[1], \var[2], \dots, \var[k]\}$ and
 $\NFA$ be an NFA over $\Alphabet \times \Vars$.
 We let $\ShortestMatching{}$ be the length of the shortest word accepted by $\NFA$, \ie{} $\ShortestMatching{} = \min_{\word \in \Lg(\NFA)} \cardinality{\word}$.
 For $m \in \Vars$,
 we let $\ShortestMatching{m}$ be the minimum number of occurrences of $\var[m]$ in the first $\ShortestMatching{}$ letters of $\word \in \Lg(\NFA)$, \ie{} $\ShortestMatching{m} = \min_{\word \in \Lg(\NFA)} \cardinality{\project{\slice{\word}{1}{\ShortestMatching{}}}{\var[m]}}$.
 For $m \in \Vars$,
 we let $\LastQS{m} \subseteq \Alphabet$ be the set of $\ShortestMatching{m}$-th letters of $\project{\Lg(\NFA)}{\var[m]}$.
 For $m \in \Vars$ and $\action \in \Alphabet$,
 we let $\SkipQS^m(\action) \in \setN$ be $\SkipQS^m(\action) = \min\big\{\ShortestMatching{m}+1, \min \{ i \in \{1,2, \dots, \ShortestMatching{m}\}\mid \exists \word \in \Lg(\NFA).\, \text{the ($\ShortestMatching{m} + 1 - i$)-th letter of $\project{\word}{\var[m]}$ is~$\action$}\} \big\}$.
\end{definition}
\begin{restatable}[correctness of QS-style skipping]{theorem}{theoremCorrectnessQSskipping}
 \label{theorem:correctness_QS}
 Let $\wordSet$ be a finite set of words,
 $\Vars = \{\var[1], \var[2], \dots, \var[k]\}$, and
 $\NAA = \tuple{\Alphabet, \Vars, \Loc, \InitLoc, \Edges, \FinalLoc}$ be an NAA.\@
 For any $\word \in \wordSet$,
 $m \in \{1, 2, \dots, k\}$,
 $i \in \{1, 2, \dots, \cardinality{\word} - \ShortestMatching{m}\}$, and
 $\tuple{(\word[1], i^1, j^1), \dots, (\word[k], i^k, j^k)} \in \Match(\NAA, \wordSet)$,
 if $\ShortestMatching{m} > 0$, $\wordi{i + \ShortestMatching{m} - 1} \not\in \LastQS{m}$, and $\word[m] = \word$, we have
 $i^m < i$ or  $i^m \geq i + \SkipQS^m(\wordi{i + \ShortestMatching{m}})$.
\end{restatable}
\begin{definition}
 [KMP-style skip values]
 Let $\NFA = \tuple{\Alphabet \times \Vars, \Loc, \InitLoc, \Edges, \FinalLoc}$ be an NFA over $\Alphabet \times \Vars$ with $\Vars = \{\var[1], \var[2], \dots, \var[k]\}$.
 For any $\loc \in \Loc$, we let $\NFA_{\loc}$ be $\NFA$ with accepting states $\{\loc\}$, \ie{} $\NFA_{\loc} = \tuple{\Alphabet \times \Vars, \Loc, \InitLoc, \Edges, \{\loc\}}$.
 For any $m \in \Vars$ and $\loc \in \Loc$, we let $\SkipKMP^m(\loc) \in \setNplus$ be 
 $\SkipKMP^m(\loc) = \min \big\{ n \in \setNplus \mid (\project{\Lg(\NFA_{\loc})}{\var[m]} \cdot \Alphabet^*) \cap (\Alphabet^n \cdot \project{\Lg(\NFA)}{\var[m]} \cdot \Alphabet^*) \neq \emptyset \big\}$.
\end{definition}
For
 an NFA $\NFA = \tuple{\Alphabet \times \Vars, \Loc, \InitLoc, \Edges, \FinalLoc}$\LongVersion{ over $\Alphabet \times \Vars$} with $\Vars = \{\var[1], \var[2], \dots, \var[k]\}$,
 $\word \in \Alphabet^*$, and
 $m \in \Vars$,
we let
$\Loc^m_{\word} \subseteq \Loc$ be the set of states reachable by a word $\extendedword \in {(\Alphabet \times \Vars)}^*$ whose $\var[m]$-projection $\project{\extendedword}{\var[m]}$ is $\word$,
\ie{} $\Loc^m_{\word} = \big\{ \loc \in \Loc \mid \word \in \project{\Lg(\NFA_{\loc})}{\var[m]} \big\}$.

\begin{restatable}[correctness of the KMP-style skipping]{theorem}{theoremKMPskipping}
 \label{theorem:correctness_KMP}
 Let $\wordSet$ be a finite set of words,
 $\Vars = \{\var[1], \var[2], \dots, \var[k]\}$, and
 $\NAA = \tuple{\Alphabet, \Vars, \Loc, \InitLoc, \Edges, \FinalLoc}$ be an NAA.\@
 For any
 $m \in \{1, 2, \dots, k\}$,
 $\word \in \wordSet$,
 $i \in \{1,2,\dots, \cardinality{\word}\}$,
 $j \geq i$, and
 $\loc \in \Loc^m_{\word}$,
 there is no
 $\tuple{(\word[1], i^1, j^1), \dots, (\word[k], i^k, j^k)} \in \Match(\NAA, \wordSet)$,
 with $\word[m] = \word$ and $i^m \in \{i + 1, i + 2, \dots, i + \SkipKMP^m(\loc) - 1\}$.
\end{restatable}
\begin{algorithm}[tb]
    \caption{An FJS-style algorithm $\algoHPMFJS$ for hyper pattern matching.}\label{algorithm:FJS}
    \footnotesize
    \KwIn{A finite set $\wordSet \subseteq \Alphabet^*$ of words and an NAA $\NAA = \tuple{\Alphabet, \Vars, \Loc, \InitLoc, \Edges, \FinalLoc}$ with $\Vars = \{\var[1], \var[2], \dots, \var[k]\}$%
 }
    \KwOut{The match set $\Match(\NAA, \wordSet) = \big\{\tuple{(\word[1], i^1, j^1), \dots, (\word[k], i^k, j^k)} \mid \tuple{\slice{\word[1]}{i^1}{j^1}, \slice{\word[2]}{i^2}{j^2}, \dots, \slice{\word[k]}{i^k}{j^k}} \in \hyperLg(\NAA) \big\}$}
    $\Match \gets \emptyset$\;
    \LongVersion{\tcp{Priority queue of the beginning of the matching trials}}
    $\waitingQueue \gets \{\tuple{i_1, \dots,i_k, \word[1], \dots, \word[k]} \mid \forall m \in \Vars.\, \word[m] \in \wordSet, 1 \leq i_m \leq \cardinality{\word[m]} - \ShortestMatching{m} + 1\}$ \;\label{algorithm:FJS:initialize}
    \While{$\cardinality{\waitingQueue} > 0$} {
        \LongVersion{\tcp{Pop the ``smallest'' element from the priority queue}}
        \KwPop{} $\tuple{i_1, \dots,i_k, \word[1], \dots, \word[k]}$ \KwFrom{} $\waitingQueue$ \;
        \If(\tcp*[f]{QS-style skipping}) {$\exists m \in \Vars.\, \wordi[m]{i_m + \ShortestMatching{m} - 1} \not\in \LastQS{m}$\label{algorithm:FJS:QS_test}} {
            \tcp{Remove the skipped starting indices}
            \For{$m \in \Vars$ satisfying $\wordi[m]{i_m + \ShortestMatching{m} - 1} \not\in \LastQS{m}$} {
                $\waitingQueue \gets \{ \tuple{i'_1, \dots,i'_k, \worda[1], \dots, \worda[k]} \in \waitingQueue \mid \forall m \in \Vars.\, \worda[m] = \word[m] \implies i'_m < i_m \lor i'_m \geq i_m + \SkipQS^m(\wordi[m]{i_m + \ShortestMatching{m}})\}$\label{algorithm:FJS:QS_remove}\;
            }
            \KwContinue{}
        }
        $p_1, p_2, \dots, p_k  \gets i_1, i_2, \dots, i_k$;\,
        $\currentConfigurations \gets \{\tuple{\wordi[1]{i_1}, \wordi[2]{i_2}, \dots, \wordi[k]{i_k}, \initLoc} \mid \initLoc \in \InitLoc\}$\;
        \LongVersion{\tcp*[f]{Start new matching trials}}
        $\reachedLoc \gets \InitLoc$ \tcp{Reached locations}
        \While{$\currentConfigurations \neq \emptyset \land \exists m \in \Vars.\, p_m < |\word[m]|$} {
            \For{$m \in \{1, 2, \dots, k\}$ satisfying $p_m < |\word[m]|$} {
            \ShortVersion{Update $\currentConfigurations$ and $\Match$
			\tcp{Same as \crefrange{algorithm:naive:setpm+1}{algorithm:naive:match} of \cref{algorithm:naive}}}
                \LongVersion{\LongVersion{\tcp{Read the $(p_m + 1)$-th letter ${\wordi[m]{p_m + 1}}$ of $\word[m]$}}
                $p_m \gets p_m + 1$\;
                \LongVersion{\tcp{Append the read letter to each configuration}}
                $\currentConfigurations \gets \{ \assign{\tuple{\worda[1], \worda[2], \dots, \worda[k], \loc}}{\worda[m]}{\worda[m] \cdot {\wordi[m]{p_m}}} \mid \tuple{\worda[1], \worda[2], \dots, \worda[k], \loc} \in \currentConfigurations\}$\;
                \For{$ \tuple{\worda[1], \worda[2], \dots, \worda[k], \loc} \in \currentConfigurations$}{
                    \LongVersion{\tcp{Apply transitions}}
                    $\currentConfigurations \gets \{\tuple{\wordb[1], \dots, \wordb[k], \loc'} \mid \tuple{\worda[1], \dots, \worda[k], \loc} \to^+ \tuple{\wordb[1], \dots, \wordb[k], \loc'}\} \cup \currentConfigurations$
                }
                \For{$ \tuple{\worda[1], \worda[2], \dots, \worda[k], \loc} \in \currentConfigurations$}{
                    \If(\tcp*[f]{Detect matching and update $\Match$}) {$\loc \in \FinalLoc$} {
                        \KwPush{} $\tuple{(\word[1], i_1, p_1 - |v_i|), \dots, (\word[k], i_k, p_k - |v_k|)}$ \KwTo{} $\Match$\;
                    }
                }}
                \LongVersion{\tcp{Update the reached locations}}
                $\reachedLoc \gets \reachedLoc \cup \{\loc \mid \tuple{\worda[1], \worda[2], \dots, \worda[k], \loc} \in \currentConfigurations\}$\;\label{algorithm:FJS:add_reached_loc}
                \LongVersion{\tcp{Remove the non-waiting configurations}}
                $\currentConfigurations \gets \{ \tuple{\worda[1], \worda[2], \dots, \worda[k], \loc} \in \currentConfigurations \mid \exists m \in \Vars.\, \worda[m] = \emptyword \}$
            }
         }
         \For(\tcp*[f]{KMP-style skipping}) {$\loc \in \reachedLoc$} {
            $\waitingQueue \gets \big\{ \tuple{i'_1, \dots,i'_k, \worda[1], \dots, \worda[k]} \in \waitingQueue \mid (\forall m \in \Vars.\, \worda[m] = \word[m]) \implies (\forall m \in \Vars.\, i'_m \leq i_m \lor i'_m \geq i_m + \SkipKMP^m(\loc)) \big\}$\;\label{algorithm:FJS:KMP_remove}
        }
    }
    \KwRet{$\Match$}\;
\end{algorithm}

\cref{algorithm:FJS} outlines our FJS-style algorithm for hyper pattern matching.
In \cref{algorithm:FJS}, the QS-style skipping is used so that:
i) we first test if the $\ShortestMatching{m}$-th letter of $\word[m]$ is in~$\LastQS{m}$ (\cref{algorithm:FJS:QS_test}) and
ii) if it is not in~$\LastQS{m}$, we remove the skipped starting indices from the waiting queue using $\SkipQS^m$ (\cref{algorithm:FJS:QS_remove}).
The KMP-style skipping is used so that:
i) we keep track of the states~$\reachedLoc$ reached during the latest matching trial (\cref{algorithm:FJS:add_reached_loc}) and
ii) we remove the unnecessary starting indices from the waiting queue using~$\SkipKMP^m$ (\cref{algorithm:FJS:KMP_remove}).
\subsection{Pruning of matching trials via projection}\label{ss:projection}

We now propose a heuristics \LongVersion{aiming at }reducing the blowup of~$\waitingQueue$ at \cref{algorithm:naive:initialize_waiting_queue} in \cref{algorithm:naive}.

\begin{example}\label{example:blowup}
	Assume a set $\wordSet = \{\word_1, \word_2, \dots, \word_k\}$ of words made of $n$ identical letters ``$\styleact{a_i}$'' ($1 \leq i \leq k$) followed by a~``$\styleact{b}$'', \ie{} $\word_i = \styleact{a_i}^n \styleact{b}$.
	Consider a pattern recognizing 3 consecutive occurrences of~``$\styleact{a_i}$'' followed by a~``$\styleact{b}$'', for each $1 \leq i \leq k$; this pattern can be easily encoded by an NAA over $k$ directions (omitted due to space concern---see \cref{appendix:example:blowup}).
	The initial $\waitingQueue$ at \cref{algorithm:naive:initialize_waiting_queue} in \cref{algorithm:naive} contains already $k^k \times (n+1)^k$ tuples; %
	even worse, each of these tuples (but one) will be explored for 3 letters in all $k$~directions until failing (technically, the tuples starting with $n+1$ and $n$, $n-1$ will explore a little less).
	But in fact, the only match will be $\Match(\NAA, \{ \word_1, \dots, \word_k \}) = \big\{ \tuple{ (\word_1, n-2, n+1) , \dots, (\word_k, n-2, n+1) } \big\}$.
\end{example}

In \cref{example:blowup}, if we knew before starting \cref{algorithm:naive} that the only potential match starts at $n-2$ in each word, we would considerably reduce this initial blowup.
This is the heuristics we propose here: instead of starting \cref{algorithm:naive} with any potential starting position, we first filter out positions in each word by keeping only these matching the pattern projected on the current direction.

\begin{definition}[Projection of an NAA]
  Given an NAA $\NAA = \tuple{\Alphabet, \Vars, \Loc, \InitLoc, \Edges, \FinalLoc}$ and given $\var \in \Vars$,
  we let $\project{\NAA}{\var}$ be the DFA %
 constructed by
  \begin{ienumerate}%
  \item replacing any transition $\tuple{a, {\var}'}$ with ``$a$'' whenever ${\var}' = \var$,
  \item replacing any transition $\tuple{a, {\var}'}$ with ``$\silentaction$'' whenever ${\var}' \neq \var$, and
  \item removing in the obtained automaton any transition labeled by~$\silentaction$, \eg{} using the powerset construction~\cite{HMU07}.
  \end{ienumerate}%
\end{definition}
\begin{example}\label{example:projection-NAA}
	Consider again $\NAA$ in \cref{figure:example}.
	We give $\project{\NAA}{\var[1]}$ and $\project{\NAA}{\var[2]}$ in \cref{figure:example-projected}.
\end{example}
\begin{figure}[tb]
	\centering
	\begin{subfigure}[b]{0.48\textwidth}
	    \centering
    \scalebox{0.8}{\begin{tikzpicture}[NFA, scale=1, yscale=1, node distance=1.5cm]
    \node[state, initial] (init) {$\loci{0}$};
    \node[state, right=of init] (p1) {$\loci{1}$};
    \node[state, right=of p1, final] (final) {$\locfinal$};

    \path [NFApath]
        (init) edge node {$\styleact{\startend}$} (p1) %
        (p1) edge[loop above] node[right] {$\styleact{a}$} (p1) %
        (p1) edge[] node[] {$\styleact{\startend}$} (final) %
    ;

    \end{tikzpicture}}
		\caption{$\project{\NAA}{\var[1]}$.}
		\label{figure:example-projected:x1}
	\end{subfigure}
	\hfill{}
	\begin{subfigure}[b]{0.48\textwidth}
    \scalebox{0.8}{\begin{tikzpicture}[NFA, scale=1, yscale=1, node distance=1.5cm]
    \node[state, initial] (init) {$\loci{0}$};
    \node[state, right=of init, final] (p1) {$\locfinal$};

    \path [NFApath]
        (init) edge node {$\styleact{\startend}$} (p1) %
        (p1) edge[loop right] node[] {$\styleact{b}$} (p1) %
    ;

    \end{tikzpicture}}
		\caption{$\project{\NAA}{\var[2]}$.}
		\label{figure:example-projected:x2}
	\end{subfigure}
	\caption{Projections of~$\NAA$ from \cref{figure:example}.}
    \label{figure:example-projected}
\end{figure}
\begin{algorithm}[tb]
    \caption{Projection heuristics for hyper pattern matching.}
    \label{algorithm:projection}

    \footnotesize

    \KwIn{A finite set $\wordSet \subseteq \Alphabet^*$ of words and an NAA $\NAA = \tuple{\Alphabet, \Vars, \Loc, \InitLoc, \Edges, \FinalLoc}$ with $\Vars = \{\var[1], \var[2], \dots, \var[k]\}$
 }
    \KwOut{Priority queue of the beginning of the matching trials}

    \For{$m \in \{1, 2, \dots, k\}$ \nllabel{algorithm:projection:for}} {
        $\Match_m \gets \algoPM \big(\wordSet, \project{\NAA}{\var[m]}\big) $\nllabel{algorithm:projection:PM}\tcp*[f]{Compute matches projected on each $\var[m]$}
    }

    \tcp{Restrict the priority queue to possible matches projected on $\var[m]$}
    $\waitingQueue \gets \big\{\tuple{i_1, \dots,i_k, \word[1], \dots, \word[k]} \mid \forall m \in \{1,\dots,k\}. \exists j_m.\, \tuple{(\word[m], i_m, j_m)} \in \Match_m \big\}$\nllabel{algorithm:projection:initialize_waiting_queue}\;

\end{algorithm}

We can now define the hyper pattern matching algorithm with projection ($\algoHPMP$) as \cref{algorithm:naive} in which we replace \cref{algorithm:naive:initialize_waiting_queue} with the algorithm fragment in \cref{algorithm:projection}.
The call to~$\algoPM$ at \cref{algorithm:projection:PM} denotes classical non-hyper pattern matching on a~DFA.
Even in the worst case, this heuristics only incurs a loss of time consisting of $k$~calls to a non-hyper pattern matching algorithm, while its gain may be an exponential decrease of the computation time in hyper pattern matching.
\begin{example}\label{example:projection-NAA-algorithm}
	Consider again $\NAA$ in \cref{figure:example}, and consider again $\word = \startend{}a\startend{}b$ from \cref{example:algorithm:naive}.
	We have $\Match_1 \gets \algoHPM\big(\{ \word \}, \project{\NAA}{\var[1]}\big) = \{ \tuple{(\word, 1, 3)} \} $
	and
	$\Match_2 \gets \{ \tuple{(\word, 3, 3)} , \tuple{(\word, 3, 4)} \} $.
	Therefore, $\waitingQueue \gets \{ \tuple{1, 3, \word, \word} \}$---a singleton to be compared to the 16~elements in the initial queue in \cref{example:algorithm:naive}.
\end{example}
\begin{algorithm}[t]
    \caption{Pruning of indices irrelevant to hyper pattern matching.}%
    \label{algorithm:filtering}
    \footnotesize
    \KwIn{A word $\word \in \Alphabet^*$ and a DFA $\project{\NAA}{\var} = \tuple{\Alphabet, \Loc, \initLoc, \Edges, \FinalLoc}$.}
    \KwOut{A word $\word[\bot] \in {(\Alphabet \cup \{\bot\})}^*$ such that $\wordi[\bot]{h} = \wordi{h}$ if there are $i$ and $j$ satisfying $i \leq h \leq j$ and $\slice{\word}{i}{j} \in \Lg(\DFA)$, and $\wordi[\bot]{h} = \bot$ otherwise.}
    $\word[\bot] \gets \emptyword$;\, $U \gets \emptyset$\;
    \tcp{$\configuration$ maintains the minimum $i$ s.t.\ we reach $\loc \in \Loc$ by using $\slice{\word}{i}{j}$}
    \KwLet{} $\configuration\colon \Loc \to \setN \cup \{\bot\}$ such that $\configuration(\loc) = \bot$ for any $\loc \in \Loc$\;
    \For{$j \in \{1, 2, \dots, \cardinality{\word}\}$ \nllabel{algorithm:filtering:for}} {
        \lIf{$\configuration(\initLoc) = \bot$} {
            $\configuration \gets \assign{\configuration}{\initLoc}{j}$
        }
        \tcp{Update $\configuration$ by applying the transition function $\Edges$}
        \KwLet{} $\configuration'\colon \Loc \to \setN \cup \{\bot\}$ such that $\configuration'(\loc) = \bot$ for any $\loc \in \Loc$\;
        \For{$\loc \in \Loc$ satisfying $\configuration(\loc) \neq \bot$} {\nllabel{algorithm:filtering:update:begin}
            $\configuration' \gets \assign{\configuration'}{\Edges(\loc, \wordi{j})}{\min\{\configuration(\loc), \configuration'(\Edges(\loc, \wordi{j}))\}}$\;
        }\nllabel{algorithm:filtering:update:end}
        $\configuration \gets \configuration'$\;
        \tcp{Update $U$ using the matching we currently have}
        \For{$\locfinal \in \FinalLoc$ satisfying $\configuration(\locfinal) \neq \bot$} {\nllabel{algorithm:filtering:final:begin}
            $U \gets U \cup \{h \mid \configuration(\locfinal) \leq h \leq j\}$\;\nllabel{algorithm:filtering:final}
        }\nllabel{algorithm:filtering:final:end}
        \tcp{Append the already determined filtering result to $\word[\bot]$}
        \For{$h \in \{\cardinality{\word[\bot]} + 1, \cardinality{\word[\bot]} + 2, \dots, \min_{\loc \in \Loc, \configuration(\loc) \neq \bot} \configuration(\loc) - 1\}$} {\nllabel{algorithm:filtering:output:begin}
            \leIf{$h \in U$}{
                $\wordi[\bot]{h} \gets \wordi{h}$
            } {
                $\wordi[\bot]{h} \gets \bot$
            }
        }\nllabel{algorithm:filtering:output:end}
    }
    \For{$h \in \{\cardinality{\word[\bot]} + 1, \cardinality{\word[\bot]} + 2, \dots, \cardinality{\word}\}$} {\nllabel{algorithm:filtering:final_output:begin}
        \leIf{$h \in U$}{
            $\wordi[\bot]{h} \gets \wordi{h}$
        } {
            $\wordi[\bot]{h} \gets \bot$
        }
    }\nllabel{algorithm:filtering:final_output:end}
    \Return{$\word[\bot]$}
\end{algorithm}

Although classical pattern matching in \cref{algorithm:projection} can be conducted efficiently, 
the overhead can be further reduced by overapproximating the exact matching.
For instance, one can use the set of indices appearing in some matching rather than the matching itself.
\cref{algorithm:filtering} shows an algorithm to identify such a set of indices.
More precisely, it maps a word $\word \in \Alphabet^*$ to another word $\word[\bot] \in {(\Alphabet \cup \{\bot\})}^*$, where the letters irrelevant to hyper pattern matching according to the projection $\DFA$ are replaced with $\bot$.
In \cref{algorithm:filtering}, for each state $\loc \in \Loc$ of the DFA $\DFA$, we maintain the minimum $i \in \setN$ such that we reach $\loc$ by feeding $\slice{\word}{i}{j}$ to $\DFA$ as a mapping $\configuration$, where $j$ is the index of the current letter to be examined.
In the loop at \crefrange{algorithm:filtering:update:begin}{algorithm:filtering:update:end}, we update $\configuration$ using the transition function $\Edges$.
In the loop at \crefrange{algorithm:filtering:final:begin}{algorithm:filtering:final:end}, we update the set $U$ of indices deemed relevant to hyper pattern matching according to the projection $\DFA$.
In the loop at \crefrange{algorithm:filtering:output:begin}{algorithm:filtering:output:end}, we update the resulting word $\word[\bot]$ in the range where the result is already determined.
Thanks to this incremental construction, we can start using the result of filtering as soon as possible.
In the loop from \crefrange{algorithm:filtering:final_output:begin}{algorithm:filtering:final_output:end}, we use the remaining part of~$U$ to update $\word[\bot]$.

The time complexity of \cref{algorithm:filtering} is linear in the length $\cardinality{\word}$ of the examined word.
In contrast, the time complexity of classical pattern matching is linear in the number of matches, which is at most $\cardinality{\word}^2$.
Since we run \cref{algorithm:filtering} for each\LongVersion{ pair} $(\word, \var) \in \wordSet \times \Vars$\LongVersion{ of word and direction},
the overall time complexity of projection-based pruning is in $\complexityClass(N \times \cardinality{\wordSet} \times \cardinality{\Vars})$, where $N$ is the maximum length of $\word \in \wordSet$.
This is more scalable than \cref{algorithm:naive}.

\section{Implementation and experiments}\label{section:implementation}
We implemented our algorithms for hyper pattern matching as a prototype tool \ourTool{}\footnote{\ourTool{} is distributed under\LongVersion{ the} GPLv3\LongVersion{ license} at \url{https://github.com/MasWag/hyppau}.}%
 in~Rust.
In particular, we implemented the following four algorithms:
the naive algorithm in \cref{algorithm:naive} ($\algoHPM$),
the algorithm with FJS-style skipping in \cref{algorithm:FJS} ($\algoHPMFJS$),
the naive algorithm with projection-based pruning with \cref{algorithm:filtering} ($\algoHPMP$), and
the algorithm with both FJS-style skipping and projection-based pruning ($\algoHPMFJSP$), \ie{} the initialization of the priority queue~$\waitingQueue$ at \cref{algorithm:FJS:initialize} in \cref{algorithm:FJS} is replaced with \cref{algorithm:projection}.
In \LongVersion{our implementation }\ourTool{}, the priority queue~$\waitingQueue$ is constructed in a lazy manner to reduce memory\LongVersion{ consumption}.
\LongVersion{%
    We conducted experiments to evaluate the efficiency of our algorithms.
}

\subsection{Benchmarks}
In our experiments, we used the following four benchmarks: \interference{}, \robustness{}, \packetPair{}, and \manydimensions{}.
All the benchmarks are our original work.
The NAAs in \interference{} and \robustness{} are shown in \cref{figure:example:interference_NFA,figure:example:robustness_NAA}, respectively.
\packetPair{} is a benchmark inspired by monitoring of network packets for data streams; see \cref{appendix:example:network} for details.
\manydimensions{} is an artificial benchmark to evaluate the scalability of \ourTool{} \wrtwithspace{}the number of directions.
The NAAs of \manydimensions{} are obtained by generalizing the NAA in \cref{figure:example}.
In all the benchmarks, we randomly generated the set of input words.

\subsection{Experiments}

We used \interference{}, \robustness{}, and \packetPair{} to observe the scalability of \ourTool{} \wrtwithspace{}the word length and the number of words, whereas
we used \manydimensions{} to observe the scalability \wrtwithspace{}the number of directions.
To observe the scalability \wrtwithspace{}the word length,
we measured the execution time \LongVersion{of \ourTool{} }using words of different lengths.
We randomly generated words of length \numrange{500}{5000}, \numrange{200}{2000}, and \numrange{1000}{10000} for \interference{}, \robustness{}, and \packetPair{}, respectively.
To observe the scalability \wrtwithspace{}the number of words,
we measured the execution time \LongVersion{of \ourTool{} }using multiple words of the same length.
We randomly generated \numrange{2}{10} words of length \numlist{500;200;1000} for \interference{}, \robustness{}, and \packetPair{}, respectively.
To observe the scalability \wrtwithspace{}the number of directions,
we measured the execution time \LongVersion{of \ourTool{} }using NAAs of different numbers of directions and words of the same length.
We generated NAAs of directions \numrange{2}{4} and randomly generated words of length~\num{200}.
We ran each of the above configurations 10 times.
We report the average execution time.
We ran all the experiments on a computing server with Intel Xeon w5-3435X \qty{4.5}{\giga\hertz} \qty{63}{\gibi\byte} RAM that runs Ubuntu 24.04.2 LTS.\@
We set the timeout to \num{1800}~seconds.

\subsection{Results and discussions}
\begin{figure}[t]
 \begin{subfigure}[b]{0.30\textwidth}
  \centering
  \scalebox{0.\ShortVersion{4}5}{\input{figures/interference-length.pgf}}
  \caption{\interference{}.}%
  \label{figure:results:length_time:interference}
 \end{subfigure}
 \hfill
 \begin{subfigure}[b]{0.30\textwidth}
  \centering
  \scalebox{0.\ShortVersion{4}5}{\input{figures/stuttering_robustness-length.pgf}}
  \caption{\robustness{}.}%
  \label{figure:results:length_time:robustness}
 \end{subfigure}
 \hfill
 \begin{subfigure}[b]{0.30\textwidth}
  \centering
  \scalebox{0.\ShortVersion{4}5}{\input{figures/network_pair-length.pgf}}
  \caption{\packetPair{}.}%
  \label{figure:results:length_time:packetpair}
 \end{subfigure}
 \caption{Elapsed time with respect to the length of the monitored words.}
 \label{figure:results:length_time}
 \begin{subfigure}[b]{0.30\textwidth}
  \centering
  \scalebox{0.\ShortVersion{4}5}{\input{figures/interference-num.pgf}}
  \caption{\interference{}.}%
  \label{figure:results:number_time:interference}
 \end{subfigure}
 \hfill
 \begin{subfigure}[b]{0.30\textwidth}
  \centering
  \scalebox{0.\ShortVersion{4}5}{\input{figures/stuttering_robustness-num.pgf}}
  \caption{\robustness{}.}%
  \label{figure:results:number_time:robustness}
 \end{subfigure}
 \hfill
 \begin{subfigure}[b]{0.30\textwidth}
  \centering
  \scalebox{0.\ShortVersion{4}5}{\input{figures/network_pair-num.pgf}}
  \caption{\packetPair{}.}%
  \label{figure:results:number_time:packetpair}
 \end{subfigure}
 \caption{Elapsed time with respect to the number of words to be monitored.}
 \label{figure:results:number_time}
\end{figure}

\cref{figure:results:length_time,figure:results:number_time} show the elapsed time with respect to the length and the number of monitored words, respectively.
In \cref{figure:results:length_time,figure:results:number_time},
we observe that for \interference{} and \robustness{}, the execution time of $\algoHPMFJS$ is slightly longer than that of $\algoHPM$, while for \packetPair{}, the execution time of $\algoHPMFJS$ is much shorter than that of $\algoHPM$.
This is because for \interference{} and \robustness{}, the skip values $\SkipKMP$ and $\SkipQS$ are at most 1, and we have no performance gain from skipping.
Due to the overhead in the use of skip values, $\algoHPMFJS$ is slightly slower than $\algoHPM$.
It is also possible to minimize this overhead by switching from $\algoHPMFJS$ to $\algoHPM$ when $\SkipKMP$ and $\SkipQS$  are 1 since we compute them beforehand.
In contrast, for \packetPair{}, the skip values $\SkipKMP$ and $\SkipQS$ are 2 for many inputs and states, and $\algoHPMFJS$ is much more efficient than $\algoHPM$ by skipping unnecessary matching trials.

In \cref{figure:results:length_time,figure:results:number_time},
we also observe a similar trend for $\algoHPM$ and $\algoHPMP$.
This is because, for \interference{} and \robustness{}, no letters in monitored words can be filtered out solely based on the projection, due to the comparison of letters observed in one word with another.
In contrast, for \packetPair{}, some letters can be filtered out because each matching for $\var[1]$ (\resp{} $\var[2]$) must start and end with $\styleactQ{s^Q}$ and $\styleactQ{e^Q}$ (\resp{} $\styleactP{s^P}$ and $\styleactP{e^P}$), and the letters outside these ranges can be filtered out.
In \cref{figure:results:number_time:packetpair}, we also observe performance gain by using both FJS-style skipping and projection-based pruning, \ie{} $\algoHPMFJSP$, compared to $\algoHPMFJS$ and $\algoHPMP$.

\begin{table}[tb]
 \caption{Elapsed time (in seconds) with respect to the number of directions for \manydimensions{}. \cellTimeout{} denotes an execution exceeding the timeout of \num{1800}~seconds.}
 \label{table:results:dimensions}
 \centering
 \footnotesize
 \begin{tabular}{lrrrr}
  \toprule
  & \algoHPM{} & \algoHPMFJS{} & \algoHPMP{} & \algoHPMFJSP{} \\
  \midrule
  $\cardinality{\Vars} = 2$ & 0.03 & 0.01 & 0.02 & 0.02 \\
  $\cardinality{\Vars} = 3$ & 4.65 & 0.95 & 1.34 & 1.12 \\
  $\cardinality{\Vars} = 4$ & \cellTimeout{} & 95.48 & 140.42 & 130.05 \\
  \bottomrule
 \end{tabular}
\end{table}
\cref{table:results:dimensions} shows the elapsed time with respect to the number of directions.
We observe that the performance gain from our heuristics is much more evident when the number of directions is large.
This is because filtering out one letter for one direction allows us to skip all the matching trials that include that letter for the direction, which also has a combinatorial explosion with respect to the number of directions.
Therefore, we conclude that our heuristics in \cref{section:heuristics} can improve the efficiency of hyper pattern matching, particularly when the number of directions is large.

Overall, although the scalability with respect to the number of directions is not good as suggested by \cref{theorem:emptiness_matching_np_complete},
one can conduct hyper pattern matching for words with thousands of letters within one minute when the number of directions is two.
Moreover, even when there are four directions, hyper pattern matching can be conducted within a few minutes for words of length 200.
Although these results may look restrictive, we believe that this is sufficient for most of the realistic use cases.
For instance, all the hyperproperties in~\cite{FRS15,BF23} bind at most two words, \ie{} can be encoded using at most two directions.

\section{Conclusions and future perspectives}\label{section:conclusion}

Toward more informative monitoring of hyperproperties, we introduced hyper pattern matching with nondeterministic asynchronous finite automata (NAAs) for representing hyperlanguages.
In addition to a naive algorithm, we developed two heuristics, FJS-style skipping and projection-based pruning, to improve its efficiency.
We evaluated the problem from both theoretical and empirical perspectives:
theoretically, we proved it is \NP{}-complete to decide the nonemptiness of the match set\LongVersion{ $\Match(\NAA, \wordSet)$};
empirically, our experimental results demonstrate that the match set can be computed for words with thousands of letters within one minute, which is likely useful for monitoring of reasonable %
 size of data.

One future direction is to generalize the problem, \eg{} to handle properties with timing constraints as in the context of \emph{timed pattern matching}~\cite{UFAM14,WHS17,BFNMA18,Waga19,WAH23}.
Another future direction is to investigate approximate algorithms\ifdefined\VersionLong, \else\footnote{\fi%
    \ie{} identifying matches within a certain threshold of errors~\cite{EGGS25}, \eg{} measured using edit distance or Hamming distance\LongVersion{,}
\ifdefined\VersionLong\else}\fi with better complexity.
\subsubsection*{Acknowledgements.}

This work is partially supported by JST PRESTO (JPMJPR22CA), JST BOOST (JPMJBY24H8), JSPS KAKENHI (22K17873), ANR BisoUS (ANR-22-CE48-0012) and ANR TAPAS (PRC ANR-24-CE25-5742).
\ifdefined\VersionLong
	\newcommand{\CCIS}{Communications in Computer and Information Science}
	\newcommand{\ENTCS}{Electronic Notes in Theoretical Computer Science}
	\newcommand{\FAC}{Formal Aspects of Computing}
	\newcommand{\FundInf}{Fundamenta Informaticae}
	\newcommand{\FMSD}{Formal Methods in System Design}
	\newcommand{\IJFCS}{International Journal of Foundations of Computer Science}
	\newcommand{\IJSSE}{International Journal of Secure Software Engineering}
	\newcommand{\IPL}{Information Processing Letters}
	\newcommand{\JAIR}{Journal of Artificial Intelligence Research}
	\newcommand{\JLAP}{Journal of Logic and Algebraic Programming}
	\newcommand{\JLAMP}{Journal of Logical and Algebraic Methods in Programming} %
	\newcommand{\JLC}{Journal of Logic and Computation}
	\newcommand{\LMCS}{Logical Methods in Computer Science}
	\newcommand{\LNCS}{Lecture Notes in Computer Science}
	\newcommand{\RESS}{Reliability Engineering \& System Safety}
	\newcommand{\RTS}{Real-Time Systems}
	\newcommand{\SCP}{Science of Computer Programming}
	\newcommand{\SOSYM}{Software and Systems Modeling ({SoSyM})}
	\newcommand{\STTT}{International Journal on Software Tools for Technology Transfer}
	\newcommand{\TCS}{Theoretical Computer Science}
	\newcommand{\TOPLAS}{{ACM} Transactions on Programming Languages and Systems ({ToPLAS})}
	\newcommand{\ToPNoC}{Transactions on {P}etri Nets and Other Models of Concurrency}
	\newcommand{\TOSEM}{{ACM} Transactions on Software Engineering and Methodology ({ToSEM})}
	\newcommand{\TSE}{{IEEE} Transactions on Software Engineering}
\else
	\newcommand{\CCIS}{CCIS}
	\newcommand{\ENTCS}{ENTCS}
	\newcommand{\FAC}{FAC}
	\newcommand{\FundInf}{FI}
	\newcommand{\FMSD}{FMSD}
	\newcommand{\IJFCS}{IJFCS}
	\newcommand{\IJSSE}{IJSSE}
	\newcommand{\IPL}{IPL}
	\newcommand{\JAIR}{JAIR}
	\newcommand{\JLAP}{JLAP}
	\newcommand{\JLAMP}{JLAMP}
	\newcommand{\JLC}{JLC}
	\newcommand{\LMCS}{LMCS}
	\newcommand{\LNCS}{LNCS}
	\newcommand{\RESS}{RESS}
	\newcommand{\RTS}{RTS}
	\newcommand{\SCP}{SCP}
	\newcommand{\SOSYM}{{SoSyM}}
	\newcommand{\STTT}{STTT}
	\newcommand{\TCS}{TCS}
	\newcommand{\TOPLAS}{ToPLAS}
	\newcommand{\ToPNoC}{ToPNOC}
	\newcommand{\TOSEM}{ToSEM}
	\newcommand{\TSE}{TSE}
\fi

\newpage

\ifdefined\VersionAuthor
	\renewcommand*{\bibfont}{\small}
	\printbibliography[title={References}]
\else
	\bibliography{biblio}
\fi

\begin{AuthorVersionBlock}
 \newpage
 \appendix

 \begin{center}
	\LARGE\bfseries
	Appendix
 \end{center}

 \section{Omitted Proofs}
 \subsection{Proof of \cref{theorem:emptiness_matching_np_complete}}\label{ss:proof:theorem:emptiness_matching_np_complete}
 \theoremMatchSetNPcomplete*

 The proof comes immediately from the following two theorems.

 \begin{restatable}{theorem}{theoremMatchSetNPhard}
 \label{theorem:emptiness_matching_np_hard}
 The nonemptiness decision problem for the match set $\Match(\NAA, \wordSet)$ for an NAA $\NAA$ and a finite set $\wordSet$ of words is \NP{}-hard.
 \end{restatable}
 \def\myqed{\qed}%
 \begin{restatable}{theorem}{theoremMatchSetNP}
 \label{theorem:emptiness_matching_np_easy}%
 The nonemptiness decision problem for the match set $\Match(\NAA, \wordSet)$ for an NAA~$\NAA$ and a finite set~$\wordSet$ of words is in~\NP{}.
 \qed{}
 \end{restatable}
 \subsubsection{Proof of \cref{theorem:emptiness_matching_np_hard}}
 \theoremMatchSetNPhard*
 \begin{proof}
 \begin{figure}[h]
  \centering
  \scalebox{0.80}{\begin{tikzpicture}[NFA, scale=1, yscale=1, node distance=2cm]
   \node[state, initial] (init) {$\loci{0}$};
   \node[state, right=of init] (p1) {$\loci{1}$};
   \node[state, right=of p1] (p2) {$\loci{2}$};
   \node[state, right=of p2,final] (p3) {$\loci{3}$};

   \path [NFApath]
   (init) edge[bend left=15] node[pos=0.5] {$\tuple{\styleact{\top}, \var[1]}$} (p1)
   (init) edge[bend right=15] node[pos=0.5,below] {$\tuple{\styleact{\bot}, \var[2]}$} (p1)
   (p1) edge[bend left=15] node {$\tuple{\styleact{\top}, \var[2]}$} (p2)
   (p1) edge[bend right=15] node[below] {$\tuple{\styleact{\bot}, \var[3]}$} (p2)
   (p2) edge[bend left=15] node {$\tuple{\styleact{\top}, \var[3]}$} (p3)
   (p2) edge[bend right=15] node[below] {$\tuple{\styleact{\bot}, \var[1]}$} (p3)
   ;
  \end{tikzpicture}}
  \caption{The NAA $\NAA$ constructed when reducing satisfiability checking of $\fml = (\prop[1] \lor \neg \prop[2]) \land (\prop[2] \lor \neg \prop[3]) \land (\prop[3] \lor \neg \prop[1])$ to nonemptiness checking of the match set.}%
  \label{figure:hpm_np_hard}
 \end{figure}
 Let $\AP = \{\prop[1], \dots, \prop[k]\}$ be a set of atomic propositions and
 let $\fml = \fml[1] \land \dots \land \fml[n]$ be a propositional formula in CNF.\@
 Let $\Vars = \{\var[1], \dots, \var[k]\}$,
 $\Alphabet = \{\top, \bot\}$\LongVersion{ (representing logical true and false)}, and
 $\wordSet = \{ \sigma^i \mid \sigma \in \Alphabet, i \leq n \}$, \ie{} the set of words of length at most $n$ containing only $\top$ or~$\bot$.
 Let $\NAA$ be the NAA over $\Alphabet$ with states $\Loc = \{\loci{0}, \dots, \loci{n}\}$
 such that there is a transition from $\loci{i-1}$ to~$\loci{i}$ labeled with $\tuple{\top, \var[j]}$ (\resp{} $\tuple{\bot, \var[j]}$) if the disjunct $\fml[i]$ contains~$\prop[j]$ (\resp{} $\neg \prop[j]$).
 Also, we let $\loci{0}$ and $\loci{n}$ be the initial and accepting states\LongVersion{, respectively}.
 See \cref{figure:hpm_np_hard} for an example.
 Since an entry of the match set is an evidence of satisfaction of~$\fml$,
 $\Match(\NAA, \wordSet)$ is nonempty if and only if ${\fml}$ is satisfiable.
 The number of transitions in $\NAA$ is the same as the number of literals in~$\fml$.
 Overall, there is a polynomial-time reduction of satisfiability checking of a propositional formula in CNF, which is \NP{}-complete~\cite{BC94},
 to nonemptiness checking of the match set.
 Thus, nonemptiness checking of the match set is \NP{}-hard.
 \qed{}
 \end{proof}
 \subsubsection{Proof of \cref{theorem:emptiness_matching_np_easy}}
 \theoremMatchSetNP*
 \begin{proof}
 [sketch]
 Let $\NAA = \tuple{\Alphabet, \Vars, \Loc, \InitLoc, \Edges, \FinalLoc}$ with
 $\Vars = \{\var[1], \var[2], \dots, \var[k]\}$.
 First, we nondeterministically pick words $\word[1], \word[2], \dots, \word[k] \in \wordSet$ and $i^1, i^2, \dots, i^k, j^1, j^2, \dots, j^k \in \setN$ satisfying $i^l \leq j^l \leq \cardinality{\word[l]}$ for each $l \in \Vars$.
 We have $\tuple{\slice{\word[1]}{i^1}{j^1}, \slice{\word[2]}{i^2}{j^2}, \dots, \slice{\word[k]}{i^k}{j^k}} \in \hyperLg(\NAA)$
 if and only if
 there is $\extendedword \in \Lg(\asNFA{\NAA})$ satisfying $\word[\var] = \project{\extendedword}{\var}$ for each $\var \in \Vars$, where $\asNFA{\NAA}$ is the underlying NFA of $\NAA$.
 Notice that we have $\cardinality{\extendedword} = \sum_{l \in \{1,2,\dots,k\}} \cardinality{\word[l]}$ for any such $\extendedword$ from the definition.
 We nondeterministically pick~$\extendedword$ satisfying $\word[\var] = \project{\extendedword}{\var}$ for each $\var \in \Vars$, and check whether $\extendedword \in \Lg(\NFA)$---which can be done in polynomial time.
 \qed{}
 \end{proof}
 \subsection{Proof of \cref{theorem:correctness_QS}}\label{ss:proof:theorem:correctness_QS}
 \theoremCorrectnessQSskipping*
 \begin{proof}
 Let $\word \in \wordSet$,
 $i \in \{1,2,\dots, \cardinality{\word}\}$,
 $\tuple{(\word[1], i^1, j^1), \dots, (\word[k], i^k, j^k)} \in \Match(\NAA, \wordSet)$, and
 $m \in \{1,2,\dots, k\}$ be such that
 $\ShortestMatching{m} > 0$,
 $\wordi{i + \ShortestMatching{m} - 1} \not\in \LastQS{m}$, and
 $\word[m] = \word$.
 From the definition of $\Match(\NAA, \wordSet)$,
 there is $\extendedword \in \Lg(\NFA)$ satisfying
 $\project{\extendedword}{\var[m]} = \slice{\word[m]}{i^m}{j^m}$.
 By $\wordi{i + \ShortestMatching{m} - 1} \not\in \LastQS{m}$,
 for any $\extendedword \in \Lg(\NFA)$,
 the $\ShortestMatching{m}$-th letter of $\slice{\word}{i}{\cardinality{\word}}$ cannot be the $\ShortestMatching{m}$-th letter
 of $\project{\extendedword}{\var[m]}$.
 Thus, we have $i \neq i_m$.
 If we have $\SkipQS^m(\wordi{i + \ShortestMatching{m}}) = 1$,
 $i \neq i_m$ immediately implies $i^m < i$ or  $i^m \geq i + \SkipQS^m(\wordi{i + \ShortestMatching{m}})$.

 Assume $\SkipQS^m(\wordi{i + \ShortestMatching{m}}) > 1$.
 From the definition of $\SkipQS^m(\wordi{i + \ShortestMatching{m}})$,
 for any $\extendedword \in \Lg(\NFA)$ and
 for any $j \in \big\{\ShortestMatching{m} + 2 - \SkipQS^m(\wordi{i + \ShortestMatching{m}}), \dots, \ShortestMatching{m}\big\}$,
 the $j$-th letter of $\project{\extendedword}{\var[m]}$ is not $\wordi{i + \ShortestMatching{m}}$.
 Therefore, we have
 $i^m \not\in \big\{i + \ShortestMatching{m} - \ShortestMatching{m} + 1, \dots, i + \ShortestMatching{m} - (\ShortestMatching{m} + 2 - \SkipQS^m(\wordi{i + \ShortestMatching{m}})) + 1 \big\} = \big\{i + 1, \dots, i + \SkipQS^m(\wordi{i + \ShortestMatching{m}}) - 1 \big\}$.
 Overall, we have $i^m < i$ or  $i^m \geq i + \SkipQS^m(\wordi{i + \ShortestMatching{m}})$.
 \qed{}
 \end{proof}
 \subsection{Proof of \cref{theorem:correctness_KMP}}\label{ss:proof:theorem:correctness_KMP}
 \begin{lemma}%
 \label{lemma:correctness_KMP}
 Let $\NFA = \tuple{\Alphabet \times \Vars, \Loc, \InitLoc, \Edges, \FinalLoc}$ be an NFA over $\Alphabet \times \Vars$ with $\Vars = \{\var[1], \var[2], \dots, \var[k]\}$.
 For any
 $\word \in \Alphabet^*$,
 $m \in \Vars$,
 $i, j \in \setN$,
 $\loc \in \Loc^m_{\slice{\word}{i}{j}}$,
 $i' \in \{i + 1, i + 2, \dots, i + \SkipKMP^m(\loc) - 1\}$, and
 $j' \geq i'$,
 we have $\slice{\word}{i'}{j'} \not\in \project{\Lg(\NFA)}{\var[m]}$.
 \end{lemma}
 \begin{proof}
 Assume $\slice{\word}{i'}{j'} \in \project{\Lg(\NFA)}{\var[m]}$.
 Let $n = i' - i$.
 By appending prefixes and suffixes to~$\slice{\word}{i'}{j'}$,
 we have $\slice{\word}{i}{\cardinality{\word}} \in \Alphabet^n \cdot \project{\Lg(\NFA)}{\var[m]} \cdot \Alphabet^*$.
 Because of $\slice{\word}{i}{\cardinality{\word}} \in \slice{\word}{i}{j} \cdot \Alphabet^*$,
 we have $\big(\slice{\word}{i}{j} \cdot \Alphabet^*\big) \cap \big(\Alphabet^n \cdot \project{\Lg(\NFA)}{\var[m]} \cdot \Alphabet^*\big) \neq \emptyset$.
 By $\slice{\word}{i}{j} \in \project{\Lg(\NFA_{\loc})}{\var[m]}$,
 we have $\big(\project{\Lg(\NFA_{\loc})}{\var[m]} \cdot \Alphabet^*\big) \cap \big(\Alphabet^n \cdot \project{\Lg(\NFA)}{\var[m]} \cdot \Alphabet^*\big) \neq \emptyset$.
 By the definition of~$\SkipKMP^m$,
 $\SkipKMP^m(\loc)$ is the minimum~$n'$ satisfying
 $\big(\project{\Lg(\NFA_{\loc})}{\var[m]} \cdot \Alphabet^*\big) \cap \big(\Alphabet^{n'} \cdot \project{\Lg(\NFA)}{\var[m]} \cdot \Alphabet^*\big) \neq \emptyset$, which contradicts $\big(\project{\Lg(\NFA_{\loc})}{\var[m]} \cdot \Alphabet^*\big) \cap \big(\Alphabet^n \cdot \project{\Lg(\NFA)}{\var[m]} \cdot \Alphabet^*\big) \neq \emptyset$ because $n \in \{1,2,\dots, \SkipKMP^m(\loc) - 1\}$.
 Thus, we have $\slice{\word}{i'}{j'} \not\in \project{\Lg(\NFA)}{\var[m]}$.
 \qed{}
 \end{proof}
 \theoremKMPskipping*
 \begin{proof}
 Assume that there is
 $\tuple{(\word[1], i^1, j^1), \dots, (\word[k], i^k, j^k)} \in \Match(\NAA, \wordSet)$,
 satisfying $\word[m] = \word$ and $i^m \in \{i + 1, i + 2, \dots, i + \SkipKMP^m(\loc) - 1\}$.
 From the definition of $\Match(\NAA, \wordSet)$,
 we have $\slice{\word}{i^m}{j^m} \in \project{\Lg(\asNFA{\NAA})}{\var[m]}$---which contradicts \cref{lemma:correctness_KMP}.
 \qed{}
 \end{proof}
 \section{Additional examples}
 \subsection{Illustrating QS-style skip values}
 \begin{example}\label{example:SM}
	Consider again the NFA~$\NAA$ in \cref{figure:example}.
	We have
		$\ShortestMatching{} = 3$,
		$\ShortestMatching{1} = 2$,
		$\ShortestMatching{2} = 1$.
	Further, $\LastQS{1} = \{ a, \startend{} \}$ while $\LastQS{2} = \{ \startend{} \}$.
	For example, $\SkipQS^1(\startend{}) = \min\big\{2+1, \min \{ i \in \{1,2\}\mid \exists \word \in \Lg(\NFA).\, \text{the ($3 - i$)-th letter of $\project{\word}{\var[1]}$ is~$\startend{}$}\} \big\} = 1$.
	Finally, $\SkipQS^1(a) = 1$ and $\SkipQS^1(b) = 3$.
 \end{example}
 \subsection{Additional example: monitoring packets over a network}\label{appendix:example:network}
 \begin{figure}[tbp]
    \centering
    \begin{tikzpicture}[NFA, scale=1, yscale=1, node distance=1.5cm]
    \node[state, initial] (init) {$\loci{0}$};
    \node[state, right=of init] (p1) {$\loci{1}$};
    \node[state, right=of p1] (p2) {$\loci{2}$};
    \node[state, above left=of p2] (p3) {$\loci{3}$};
    \node[state, above right=of p2] (p4) {$\loci{4}$};
    \node[state, right=of p2] (p5) {$\loci{5}$};
    \node[state, right=of p5, final] (final) {$\locfinal$};

    \path [NFApath]
        (init) edge node[below] {$\tuple{\styleactQ{s^Q}, \var[1]}$} (p1)
        (p1) edge node[below] {$\tuple{\styleactP{s^P}, \var[2]}$} (p2)
        (p2) edge[] node[yshift=0.9em] {$\tuple{\styleactQ{Q}, \var[1]}$} (p3)
        (p3) edge[] node[] {$\tuple{\styleactP{P}, \var[2]}$} (p4)
        (p4) edge[] node[yshift=0.9em] {$\tuple{\styleactP{P}, \var[2]}$} (p2)
        (p2) edge node[below] {$\tuple{\styleactQ{e^Q}, \var[1]}$} (p5)
        (p5) edge node[below] {$\tuple{\styleactP{e^P}, \var[2]}$} (final)
    ;
    \end{tikzpicture}
    \caption{Example: matching requests and responses.}%
    \label{figure:example-requests-responses}
 \end{figure}
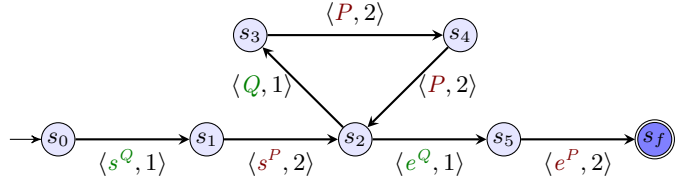
 \begin{example}\label{example:network}
    Assume a system monitors a network with requests and responses; a server is processing data streams (requests) such as video or audio streams and sends them back after processing (\eg{} transcoding or real-time decompression), not necessarily in the order they were sent.
    Notice that such an assumption is common in UDP-based streaming protocols, \eg{} Real-time Transport Protocol (RTP)~\cite{rfc1889}.
    The processing doubles the size of the requests.
    Requests (resp.~responses) are made of a start packet~$\styleactQ{s^Q}$ (resp.~$\styleactP{s^P}$), an arbitrary list of data packets~$\styleactQ{Q}$ (resp.~$\styleactP{P}$), and an end packet~$\styleactQ{e^Q}$ (resp.~$\styleactP{s^P}$).
    A possible log~$\word$ is

    \newcommand{\styleIndex}[1]{\color{gray}\scalebox{.6}{#1}}

    \[
    \begin{array}{*{30}{@{\,} c}}
    	\styleIndex{1} & \styleIndex{2} & \styleIndex{3} & \styleIndex{4} & \styleIndex{5} & \styleIndex{6} & \styleIndex{7} & \styleIndex{8} & \styleIndex{9} & \styleIndex{10} &
    	\styleIndex{11} & \styleIndex{12} & \styleIndex{13} & \styleIndex{14} & \styleIndex{15} & \styleIndex{16} & \styleIndex{17} & \styleIndex{18} & \styleIndex{19} & \styleIndex{20} &
    	\styleIndex{21} & \styleIndex{22} & \styleIndex{23} & \styleIndex{24} & \styleIndex{25} & \styleIndex{26} & \styleIndex{27} & \styleIndex{28} & \styleIndex{29} & \styleIndex{30}
    	\\
    	\styleactQ{s^Q}  & \styleactQ{Q}  & \styleactQ{e^Q}  & \styleactQ{s^Q}  & \styleactQ{Q}  & \styleactQ{Q}  & \styleactP{s^P}   & \styleactQ{e^Q}  & \styleactP{P}  & \styleactP{P}  & \styleactQ{s^Q}  & \styleactQ{Q}  & \styleactP{P}  & \styleactQ{Q}  & \styleactP{P}  & \styleactP{e^P}  & \styleactP{s^P}  & \styleactQ{Q}  & \styleactP{P}  & \styleactP{P}  & \styleactP{P}  & \styleactP{P}  & \styleactP{P}  & \styleactQ{e^Q}  & \styleactP{P}  & \styleactP{e^P}  & \styleactP{s^P}  & \styleactP{P}  & \styleactP{P}  & \styleactP{e^P}
    \end{array}
    \]

    Assuming all requests have a different size, the NAA in \cref{figure:example-requests-responses} allows to ``deanonymize'' responses, by matching them with the requests, such that the responses have a size twice as large as the request, where the self-loops to ignore irrelevant letters at $\loci{2}, \loci{3}, \loci{4}$ are omitted.
    (If sizes are not unique, then we get all \emph{potential} such matches.)
    In the aforementioned log, three requests of growing size (size~1 then~2 then~3) are sent; the match is
    \[  \big\{ \tuple{ (\word, 1, 3) , (\word, 27, 30) } , \tuple{ (\word, 4, 8) , (\word, 7, 16) , \tuple{ (\word, 11, 24) , (\word, 17, 26)  }} \big\} \text{,} \]
    \noindent{}\ie{} we can deduce that the server first processes the 2nd request (even though it is not yet fully sent, \ie{} the request is interleaved with the processed response), then the 3rd one, then the first one.
 \end{example}
 \section{Additional details for \cref{example:blowup}}\label{appendix:example:blowup}

 See \cref{figure:blowup:NAA} for the NAA corresponding to \cref{example:blowup}.

 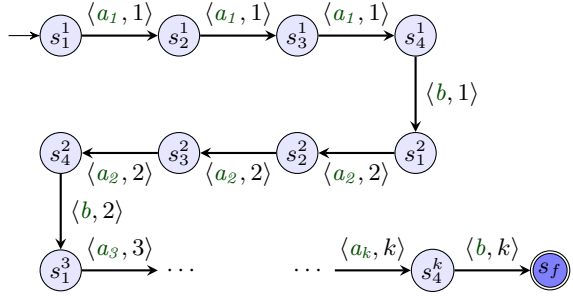
\begin{figure}[h]
    \centering
    \begin{tikzpicture}[NFA, scale=1, yscale=.5, node distance=1cm]
    \node[state, initial] (p1_1) {$\loc_1^1$};
    \node[state, right=of p1_1] (p2_1) {$\loc_2^1$};
    \node[state, right=of p2_1] (p3_1) {$\loc_3^1$};
    \node[state, right=of p3_1] (p4_1) {$\loc_4^1$};

    \node[state, below=of p4_1] (p1_2) {$\loc_1^2$};
    \node[state, left=of p1_2] (p2_2) {$\loc_2^2$};
    \node[state, left=of p2_2] (p3_2) {$\loc_3^2$};
    \node[state, left=of p3_2] (p4_2) {$\loc_4^2$};

    \node[state, below=of p4_2] (p1_3) {$\loc_1^3$};
    \node[draw=none, right=of p1_3] (p2_3) {$\cdots$};
    \node[draw=none, right=of p2_3] (p3_k) {$\cdots$};
    \node[state, right=of p3_k] (p4_k) {$\loc_4^k$};
    \node[state, final, right=of p4_k] (final) {$\locfinal$};

    \path [NFApath]
        (p1_1) edge node {$\tuple{\styleact{a_1}, \var[1]}$} (p2_1)
        (p2_1) edge node {$\tuple{\styleact{a_1}, \var[1]}$} (p3_1)
        (p3_1) edge node {$\tuple{\styleact{a_1}, \var[1]}$} (p4_1)
        (p4_1) edge node {$\tuple{\styleact{b}, \var[1]}$} (p1_2)

        (p1_2) edge node {$\tuple{\styleact{a_2}, \var[2]}$} (p2_2)
        (p2_2) edge node {$\tuple{\styleact{a_2}, \var[2]}$} (p3_2)
        (p3_2) edge node {$\tuple{\styleact{a_2}, \var[2]}$} (p4_2)
        (p4_2) edge node {$\tuple{\styleact{b}, \var[2]}$} (p1_3)

        (p1_3) edge node {$\tuple{\styleact{a_3}, \var[3]}$} (p2_3)
        (p3_k) edge node {$\tuple{\styleact{a_k}, \var[k]}$} (p4_k)
        (p4_k) edge node {$\tuple{\styleact{b}, \var[k]}$} (final)

    ;
    \end{tikzpicture}
    \caption{NAA over $k$ directions recognizing 3 consecutive occurrences of~$a_i$ followed by a~$b$, for each $1 \leq i \leq k$.}%
    \label{figure:blowup:NAA}
 \end{figure}
 \section{Details of our experiments}\label{appendix:detail_experiments}
 \subsection{Detailed analysis of the scalability}
 \begin{figure}[tb]
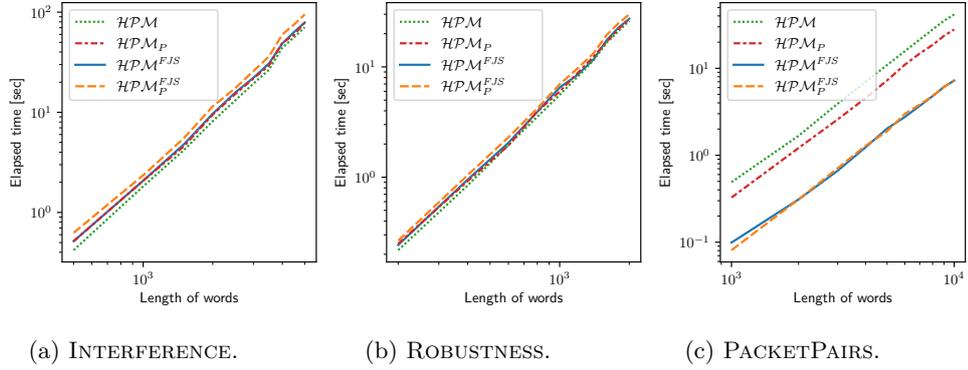

 \begin{subfigure}[b]{0.30\textwidth}
  \centering
  \scalebox{0.55}{\input{figures/interference-length-log-log.pgf}}
  \caption{\interference{}.}%
  \label{figure:results:length_time:log_log:interference}
 \end{subfigure}
 \hfill
 \begin{subfigure}[b]{0.30\textwidth}
  \centering
  \scalebox{0.55}{\input{figures/stuttering_robustness-length-log-log.pgf}}
  \caption{\robustness{}.}%
  \label{figure:results:length_time:log_log:robustness}
 \end{subfigure}
 \hfill
 \begin{subfigure}[b]{0.30\textwidth}
  \centering
  \scalebox{0.55}{\input{figures/network_pair-length-log-log.pgf}}
  \caption{\packetPair{}.}%
  \label{figure:results:length_time:log_log:packetpair}
 \end{subfigure}
 \caption{Elapsed time with respect to the length of the monitored words on a log-log scale.}
 \label{figure:results:length_time:log_log}
 \end{figure}
 \begin{figure}[tb]
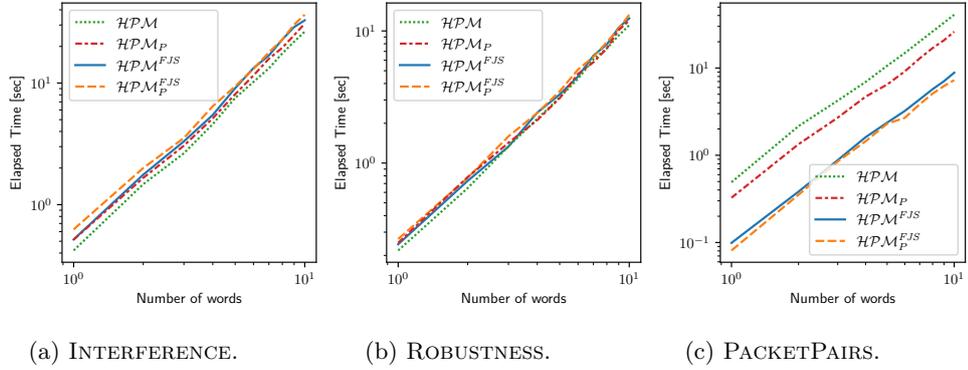

 \begin{subfigure}[b]{0.30\textwidth}
  \centering
  \scalebox{0.55}{\input{figures/interference-num-log-log.pgf}}
  \caption{\interference{}.}%
  \label{figure:results:number_time:log_log:interference}
 \end{subfigure}
 \hfill
 \begin{subfigure}[b]{0.30\textwidth}
  \centering
  \scalebox{0.55}{\input{figures/stuttering_robustness-num-log-log.pgf}}
  \caption{\robustness{}.}%
  \label{figure:results:number_time:log_log:robustness}
 \end{subfigure}
 \hfill
 \begin{subfigure}[b]{0.30\textwidth}
  \centering
  \scalebox{0.55}{\input{figures/network_pair-num-log-log.pgf}}
  \caption{\packetPair{}.}%
  \label{figure:results:number_time:log_log:packetpair}
 \end{subfigure}
 \caption{Elapsed time with respect to the number of words to be monitored on a log-log scale.}
 \label{figure:results:number_time:log_log}
 \end{figure}

 \cref{figure:results:length_time:log_log,figure:results:number_time:log_log} show the execution time of \ourTool{} with respect to the length of monitored words and the number of monitored words respectively, on a log-log scale.
 We observe that for all the benchmarks and methods, the plot on a log-log scale is approximately linear, suggesting that the execution time is polynomial with respect to the length of the monitored words.
 For each benchmark, we performed a linear regression on the log-log scale.
 The estimated slope was between 1.6 and~2.3.
 This suggests that the execution time of \ourTool{} scales approximately quadratically with respect to the length and the number of monitored words.
 This is because the initial size of the priority queue $\waitingQueue$ in \cref{algorithm:naive:initialize_waiting_queue} of \cref{algorithm:naive} increases quadratically when we have two word variables.

 \subsection{Cost of skip-value computation}
 \begin{table}[tbp]
 \centering
 \caption{Time (in \unit{\micro\second}) to construct the skip values in \ourTool{} for each of the benchmarks.}%
 \label{table:skip_value_cost}
 \begin{tabular}{lrrrr}
  \toprule
  Benchmark & \multicolumn{2}{c}{KMP} & \multicolumn{2}{c}{QS} \\
  & \algoHPMFJS{} & \algoHPMFJSP{} & \algoHPMFJS{} & \algoHPMFJSP{} \\
  \midrule
  \interference{} & 8285.60 & 8245.19 & 1481.16 & 1449.04 \\
  \robustness{} & 27189.48 & 27773.85 & 781.82 & 790.17 \\
  \packetPair{} & 4015.16 & 4005.07 & 2201.33 & 2181.25 \\
  (\manydimensions{}, 2) & 110.82 & 136.16 & 21.39 & 23.63 \\
  (\manydimensions{}, 3) & 336.50 & 322.86 & 32.42 & 26.03 \\
  (\manydimensions{}, 4) & 650.56 & 434.57 & 36.41 & 31.14 \\
  \bottomrule
 \end{tabular}
 \end{table}

 In \cref{table:skip_value_cost}, we observe that the time of the pre-computation for the skip values depends on the size of the NAA.\@
 For instance, the pre-computation of $\SkipKMP$ took about 27 milliseconds for \robustness{}, whereas it only took less than 0.1 milliseconds for (\manydimensions{}, 2).
 Nevertheless, compared with the cost of hyper pattern matching, this overhead is very small.
 Therefore, we conclude that the overhead of the pre-computation in $\algoHPMFJS$ is ignorable.
\end{AuthorVersionBlock}

\end{document}